\documentclass{lmcs}
\pdfoutput=1

\usepackage{lastpage}
\lmcsdoi{15}{2}{16}
\lmcsheading{}{\pageref{LastPage}}{}{}%
{May~29,~2017}{May~23,~2019}{}

\usepackage{tikz}

\usepackage{amsthm}
\usepackage{amsmath}
\usepackage{amsfonts}
\usepackage{amssymb}
\usepackage[shortcuts]{extdash}

\newcommand{\newthm}[2]{\newtheorem{#1}[thm]{#2}}

\theoremstyle{plain}
\newtheorem{theorem}[thm]{Theorem}
  \newthm{lemma}{Lemma}
  \newthm{corollary}{Corollary}
  \newthm{proposition}{Proposition}
\theoremstyle{definition}
  \newthm{definition}{Definition}
  \newthm{problem}{Problem}
  \newthm{example}{Example}
  \newthm{remark}{Remark}

\usepackage{xspace}

\usetikzlibrary{automata,positioning,calc,arrows}
\usetikzlibrary{arrows.meta}

\tikzset{every picture/.append style={initial text={}}}

\tikzset{every loop/.append style={->,-{Latex[length=2.5mm]}}}

\tikzset{every initial by arrow/.style={->,-{Latex[length=2.5mm]}, initial distance=1.3}}

\tikzstyle{trans}=[draw,-{Latex[length=2.5mm]},auto]

\tikzstyle{dot} = [draw,shape=circle,fill, minimum size=1.5mm, inner sep=0pt,outer sep=0pt]

\tikzstyle{acc} = [draw=black]
\tikzstyle{tra} = [draw=black, dashed]

\tikzstyle{idf} = [scale=0.8, pos=0.65]

\tikzset{every picture/.append style={initial text={}}}

\newcommand{\dxs}{2}

\newtheorem{claim}[thm]{Claim}
\newtheorem{subclaim}[thm]{Subclaim}

\newcommand{\treeShapedQDag}{

  \foreach \y in {0,...,3} {
    \foreach \x in {0,...,7} {
     	\node[dot] at (\dxs*\x, -\y) {};
    }
    
    \node at (\dxs*7.5, -\y) {$\cdots$};
  }
  
  \path[tra] (0*\dxs, -0) -- ++(\dxs, -0);
  \path[acc] (0*\dxs, -0) -- ++(\dxs, -1);
  
  \path[acc] (1*\dxs, -0) -- ++(\dxs, -0);
  \path[acc] (1*\dxs, -0) -- ++(\dxs, -1);
  \path[acc] (1*\dxs, -1) -- ++(\dxs, -1);
  
  \path[acc] (2*\dxs, -0) -- ++(\dxs, -0);
  \path[tra] (2*\dxs, -0) -- ++(\dxs, -1);
  \path[tra] (2*\dxs, -0) -- ++(\dxs, -2);
  \path[acc] (2*\dxs, -1) -- ++(\dxs, -2);
  
  \path[acc] (3*\dxs, -0) -- ++(\dxs, -0);
  \path[tra] (3*\dxs, -2) -- ++(\dxs, -0);
  \path[acc] (3*\dxs, -2) -- ++(\dxs, -1);
  
  \path[acc] (4*\dxs, -0) -- ++(\dxs, -0);
  \path[acc] (4*\dxs, -2) -- ++(\dxs, -0);
  
  \path[tra] (5*\dxs, -2) -- ++(\dxs, +1);
  \path[acc] (5*\dxs, -2) -- ++(\dxs, -1);
  \path[acc] (5*\dxs, -2) -- ++(\dxs, +2);
  
  \path[tra] (6*\dxs, -0) -- ++(\dxs, -2);
  \path[tra] (6*\dxs, -1) -- ++(\dxs, +1);
  \path[acc] (6*\dxs, -1) -- ++(\dxs, -0);
  \path[acc] (6*\dxs, -3) -- ++(\dxs, -0);
}

\newcommand{\qDagFigure}{

  \foreach \y in {0,...,3} {
    \foreach \x in {0,...,7} {
     	\node[dot] at (\dxs*\x, -\y) {};
    }
    
    \node at (\dxs*7.5, -\y) {$\cdots$};
  }
  
  \path[tra] (0*\dxs, -0) -- ++(\dxs, -0);
  \path[acc] (0*\dxs, -0) -- ++(\dxs, -1);
  
  \path[acc] (1*\dxs, -0) -- ++(\dxs, -0);
  \path[acc] (1*\dxs, -0) -- ++(\dxs, -1);
  \path[tra] (1*\dxs, -1) -- ++(\dxs, -0);
  \path[acc] (1*\dxs, -1) -- ++(\dxs, -1);
  
  \path[acc] (2*\dxs, -0) -- ++(\dxs, -0);
  \path[tra] (2*\dxs, -0) -- ++(\dxs, -1);
  \path[tra] (2*\dxs, -0) -- ++(\dxs, -2);
  \path[acc] (2*\dxs, -1) -- ++(\dxs, -0);
  \path[acc] (2*\dxs, -1) -- ++(\dxs, -2);
  \path[acc] (2*\dxs, -2) -- ++(\dxs, -0);
  
  \path[acc] (3*\dxs, -0) -- ++(\dxs, -0);
  \path[tra] (3*\dxs, -1) -- ++(\dxs, +1);
  \path[tra] (3*\dxs, -2) -- ++(\dxs, -0);
  \path[acc] (3*\dxs, -2) -- ++(\dxs, -1);
  \path[acc] (3*\dxs, -3) -- ++(\dxs, -0);
  
  \path[acc] (4*\dxs, -0) -- ++(\dxs, -0);
  \path[tra] (4*\dxs, -0) -- ++(\dxs, -2);
  
  \path[acc] (4*\dxs, -2) -- ++(\dxs, -0);
  \path[tra] (4*\dxs, -3) -- ++(\dxs, +1);
  
  \path[acc] (5*\dxs, -0) -- ++(\dxs, -0);
  \path[tra] (5*\dxs, -0) -- ++(\dxs, -1);
  \path[tra] (5*\dxs, -2) -- ++(\dxs, +1);
  \path[acc] (5*\dxs, -2) -- ++(\dxs, -1);
  \path[acc] (5*\dxs, -2) -- ++(\dxs, +2);
  
  \path[acc] (6*\dxs, -0) -- ++(\dxs, -0);
  \path[tra] (6*\dxs, -0) -- ++(\dxs, -2);
  \path[tra] (6*\dxs, -1) -- ++(\dxs, +1);
  \path[acc] (6*\dxs, -1) -- ++(\dxs, -0);
  \path[acc] (6*\dxs, -3) -- ++(\dxs, -0);
  \path[acc] (6*\dxs, -3) -- ++(\dxs, +1);
    
  \newcommand{\hors}{18}
  
  \foreach \x in {0,...,1} {
  	\foreach \y in {0,...,3} {
     	\node[dot] at (\hors + \dxs*\x, -\y) {};
    }
  }
  
  \path[tra] (\hors + 0*\dxs, -0) -- ++(\dxs, -0);
  \path[acc] (\hors + 0*\dxs, -0) -- ++(\dxs, -1);
  \path[acc] (\hors + 0*\dxs, -0) -- ++(\dxs, -3);
  \path[acc] (\hors + 0*\dxs, -1) -- ++(\dxs, -2);
  \path[tra] (\hors + 0*\dxs, -2) -- ++(\dxs, +1);
  \path[tra] (\hors + 0*\dxs, -2) -- ++(\dxs, -1);
}

\newenvironment{proofof}[1]
	{\vspace{1ex}\noindent{\it Proof of #1.}\hspace{0.5em}} 
    {\hfill\qed\vspace{1ex}}

\newcommand{\comment}[1]{}

\newcommand{\init}{\mathrm{I}}

\newcommand{\infpow}{\infty}

\newcommand{\trNon}{{0}}
\newcommand{\trYes}{{1}}
\newcommand{\trAcc}{{\star}}

\newcommand{\fun}[3]{#1\colon #2\to #3}

\newcommand{\trans}[1]{\overset{#1}{\longrightarrow}}

\newcommand{\infA}{\alpha}
\newcommand{\infB}{\beta}
\newcommand{\patA}{\pi}
\newcommand{\colA}{C}
\newcommand{\homA}{h}

\newcommand{\mathcalsym}[1]{\ensuremath{\mathcal{#1}}\xspace}

\newcommand{\Aa}{\mathcalsym{A}}
\newcommand{\Bb}{\mathcalsym{B}}

\newcommand{\Jj}{\mathcalsym{J}}

\newcommand{\Rr}{\mathcalsym{R}}

\newcommand{\Tt}{\mathcalsym{T}}

\theoremstyle{plain}
\newtheorem*{rep@theorem}{\rep@title}
\newcommand{\newreptheorem}[2]{%
\newenvironment{rep#1}[1]{%
\def\rep@title{#2 \ref{##1}}%
\begin{rep@theorem}}%
{\end{rep@theorem}}}

\newreptheorem{theorem}{Theorem}
\newreptheorem{lemma}{Lemma}
\newreptheorem{proposition}{Proposition}
\newreptheorem{conjecture}{Conjecture}
\newreptheorem{claim}{Claim}

\newcommand{\ind}{\textrm{-}\mathsf{IND}}

\newcommand{\rca}{\mathsf{RCA}_0}
\newcommand{\wkl}{\mathsf{WKL}_0}
\newcommand{\rt}{\mathsf{RT}}
\newcommand{\srt}{\mathsf{SRT}}

\newcommand{\tuple}[1]{\langle #1 \rangle}

\newcommand{\mso}{\ensuremath{\mathrm{MSO}}\xspace}
\newcommand{\GrJ}{\mathrel{\mathcal{J}}}
\newcommand{\GrH}{\mathrel{\mathcal{H}}}
\newcommand{\GrL}{\mathrel{\mathcal{L}}}
\newcommand{\GrR}{\mathrel{\mathcal{R}}}

\usepackage{stmaryrd}
\newcommand{\lbracket}{{\langle}}
\newcommand{\rbracket}{{\rangle}}

\newcommand{\godelize}[1]{\lceil #1 \rceil}

\newcommand{\bnfalt}{\; | \;}
\newcommand{\Sing}{\texttt{Sing}}

\newcommand{\eqdef}{\mathrel{\mathop:}=}

\usepackage{todonotes}
\newcommand{\Nn}{\mathbb{N}}

\keywords{nondeterministic automata, monadic second-order logic, B\"uchi's theorem, additive Ramsey's theorem, reverse mathematics}

\begin{document}

\title[The logical strength of~B\"uchi's decidability theorem]{The logical strength of~B\"uchi's decidability theorem}
\titlecomment{{\lsuper*}This article is an~extended journal version of an equally named paper from CSL 2016.}

\author[L.~A.~Ko{\l}odziejczyk]{Leszek A. Ko{\l}odziejczyk\rsuper{a}}	
\address{\lsuper{a}University of Warsaw, Faculty of Mathematics, Informatics, and Mechanics\\
  Banacha 2, 02-097 Warsaw, Poland}	
\email{lak@mimuw.edu.pl}  
\email{henrykm@mimuw.edu.pl}
\email{mskrzypczak@mimuw.edu.pl}
\thanks{The first author has been partially supported by Polish National Science Centre grant no.\ 2013/09/B/ST1/04390; the remaining three authors have been supported by Polish National Science Centre grant no.\ 2014-13/B/ST6/03595.}

\author[H.~Michalewski]{Henryk Michalewski\rsuper{a}}	

\author[C.~Pradic]{C\'ecilia Pradic\rsuper{{a,b}}}	
\address{\lsuper{b}ENS Lyon, 46 all\'ee d'Italie, 69007 Lyon, France}	

\author[M.~Skrzypczak]{Micha{\l} Skrzypczak\rsuper{a}}	

\begin{abstract}
We study the strength of axioms needed to prove various results related to automata on infinite words and B\"uchi's theorem on the decidability of the \mso theory of $(\Nn, {\le})$. We prove that the following are equivalent over the weak second-order arithmetic theory $\rca$:

\begin{enumerate}
\item the induction scheme for $\Sigma^0_2$ formulae of arithmetic,
\item a variant of Ramsey's Theorem for pairs restricted to so\=/called \emph{additive colourings},
\item B\"uchi's complementation theorem for nondeterministic automata on infinite words,
\item the decidability of the depth-$n$ fragment of the \mso theory of $(\Nn, {\le})$, for each $n \ge 5$,
\end{enumerate}

Moreover, each of (1)--(4) implies McNaughton's determinisation theorem for automata on infinite words, as well as the ``bounded-width'' version of K\"onig's Lemma, often used in proofs of McNaughton's theorem.
\end{abstract}

\maketitle

\section{Introduction}

B\"uchi's theorem~\cite{buchi1962} states that the monadic second-order theory of $(\Nn, {\le})$ is decidable. This is one of the  fundamental results on the decidability of logical theories, and no less fundamental are the methods developed in order to prove it.

Typical proofs of B\"uchi's theorem  make use of automata on infinite words. B\"uchi's original argument involved obtaining a complementation theorem for nondeterministic word automata: for each such automaton $\Aa$, there is another automaton $\Bb$ which accepts a given word exactly if $\Aa$ does not. Thanks to the complementation theorem, an \mso formula can be inductively translated into an equivalent nondeterministic automaton. At that point, checking satisfiability of the formula becomes a matter of elementary combinatorics. Another approach to decidability of \mso was presented by Shelah in~\cite{shelah1975}. Shelah's ``composition method'' is automata-free, but is similar to B\"uchi's proof in one important respect: both use a restricted form of Ramsey's Theorem. 

McNaughton~\cite{mcn1966} showed that an infinite word automaton can be determinised, though at the cost of allowing automata with a more general acceptance condition than B\"uchi's.
Since deterministic automata are easy to complement, this again gives the translation of formulae to automata and thus decidability of \mso. To the best of our knowledge all determinisation proofs known from the literature rely on either a restricted form of Ramsey's Theorem or a restricted form of K\"onig's Lemma. The former appears for example in the algebraic approach described in~\cite{PinPerrin}; the latter is used here, see Section~\ref{sec:ind-to-det}.


It is natural to ask how the various proofs of B\"uchi's theorem and related results compare to one another. For instance, is determinisation of word automata an ``essentially stronger'' result than complementation? Also, is the use of mildly nonconstructive principles \`a la Ramsey or K\"onig unavoidable?

A convenient framework for studying questions of this sort is provided by the programme of \emph{reverse mathematics}~\cite{simpson}. The idea is to compare various theorems as formalised in the very expressive language of an axiomatic theory known as \emph{second-order arithmetic}. Typical subtheories of second-order arithmetic are axiomatised by principles asserting the existence of more or less complicated sets of natural numbers. An important example is the relatively weak theory $\rca$, which guarantees only the existence of decidable sets. $\rca$ can formalise a significant amount of everyday mathematics and prove the termination of any primitive recursive algorithm, but it is unable to prove the existence of noncomputable objects such as the homogeneous sets postulated by Ramsey's Theorem or the infinite branches postulated by K\"onig's Lemma. Sometimes it is possible to show that two theorems not provable in $\rca$ are provably equivalent in it, and thus neither theorem is logically stronger than the other in the sense of requiring more abstract or less constructive sets. It is also often the case that a set existence principle used to derive some theorem is actually implied by the theorem over $\rca$. 
This serves as evidence that the principle is in fact necessary to prove the theorem. 

In this paper, we carry out a reverse-mathematical study of the results around B\"uchi's theorem. 
We have two main aims in mind. One is to compare complementation, determinisation, and decidability of \mso in terms of logical strength. The other aim is to clarify the role of Ramsey's Theorem and K\"onig's Lemma in proofs of B\"uchi's theorem and the related facts about automata. This seems interesting in light of the fact that the usual formulation of Ramsey's Theorem for pairs and the so-called Weak K\"onig's Lemma (the form of K\"onig's Lemma most commonly needed in practice) are known to be incomparable over $\rca$~\cite{hirst:thesis, liu}.

Our findings are as follows: firstly, determinisation of infinite word automata is no stronger than complementation, at least in the sense of implication over $\rca$. Secondly, decidability of \mso over $(\Nn,{\le})$ implies both complementation and determinisation. 
Finally, the use of Ramsey- or K\"onig-like principles in proofs of B\"uchi's theorem is mostly spurious in the sense that the versions that are actually needed follow from a very limited set-existence principle, namely mathematical induction for properties expressed by $\Sigma^0_2$ formulae. More precisely, we prove:

\pagebreak
\begin{theorem}
\label{thm:main}
Over $\rca$, the following statements are equivalent:
\begin{enumerate}[align=left,ref=(\arabic*)]
\item the principle of mathematical induction for properties expressed using $\Sigma^0_2$ formulae (denoted $\Sigma^0_2\ind$),
\label{it:induction}
\item Additive Ramsey's Theorem (see Definition~\ref{def:additive-ramsey}),
\label{it:ramsey}
\item effective complementation for B\"uchi automata: there exists an algorithm which for each nondeterministic B\"uchi automaton $\Aa$ outputs a B\"uchi automaton $\Bb$ such that for every infinite word $\infA$, $\Bb$ accepts $\infA$ exactly if $\Aa$ does not accept $\infA$,
\label{it:compl_buchi}
\item[(3')] complementation for B\"uchi automata: for every nondeterministic B\"uchi automaton $\Aa$ there exists a B\"uchi automaton $\Bb$ such that for every infinite word $\infA$, $\Bb$ accepts $\infA$ exactly if $\Aa$ does not accept $\infA$,
\item the decidability of the depth-$n$ fragment of the $\mso$ theory of $(\Nn, {\le})$ (where $n \geq 5$ is a natural number)\footnote{The restriction to fixed-depth fragments is a technicality related to undefinability of truth. This is explained in more detail in Section~\ref{sec:intro_automata}.}.
\label{it:dec_mso}
\end{enumerate}
Furthermore, each of~\ref{it:induction}--\ref{it:dec_mso} implies:
\begin{enumerate}[resume*]
\item determinisation of B\"uchi automata: there exists an algorithm which for each nondeterministic B\"uchi automaton $\Aa$
outputs a deterministic Rabin automaton $\Bb$ such that for every infinite word $\infA$, $\Bb$ accepts $\infA$ exactly if $\Aa$ accepts $\infA$.
\label{it:det_buchi}
\end{enumerate}
\end{theorem}

\newcommand{\refComplBuchiPrime}{{(3')}\xspace}

We also give a precise statement of the bounded-width form of K\"onig's Lemma often used in proofs of~\ref{it:det_buchi}, and show that it is implied by each of~\ref{it:induction}--\ref{it:dec_mso}. Interestingly, it is not clear if~\ref{it:det_buchi} implies~\ref{it:induction}--\ref{it:dec_mso} over $\rca$: standard arguments used to complement deterministic automata with acceptance conditions other than B\"uchi seem to involve $\Sigma^0_2\ind$.
%

\begin{remark}
Our proof of the implication $\ref{it:ramsey}\to \ref{it:compl_buchi}$ 
actually shows that Additive Ramsey's Theorem implies the correctness of the standard complementation algorithm known from the literature. In the light of the equivalence between $\ref{it:compl_buchi}$ and $\refComplBuchiPrime$, this means that the correctness of the standard algorithm is already implied by the seemingly weaker statement that complementation of a~B\"uchi automaton is always possible. Also, the algorithm we use in our proof of~\ref{it:det_buchi} is a typical one.
\end{remark}

It follows from our results that B\"uchi's theorem is unprovable in $\rca$, but only barely: it is true in computable mathematics, in the sense that the theorem remains valid if all the set quantifiers are restricted to range over (exactly) the decidable subsets of $\Nn$. This is in stark contrast to the behaviour of Rabin's theorem on the decidability of \mso on the infinite binary tree, which is known to require the existence of extremely complicated noncomputable sets~\cite{km:2016}. Also Additive Ramsey's Theorem and Bounded-width K\"onig's Lemma are true in computable mathematics---quite unlike more general forms of Ramsey's Theorem for pairs and K\"onig's Lemma~\cite{jockusch:ramsey,kreisel:variant}.

To prove the implication $\ref{it:dec_mso} \to \ref{it:induction}$ of Theorem~\ref{thm:main}, we come up with a family of \mso sentences for which truth in $(\Nn,{\le})$ is undecidable if $\Sigma^0_2\ind$ fails. The other implications are proved by formalising more or less standard arguments from automata theory. In some cases this is routine, but especially the proof of $\ref{it:induction}\to \ref{it:det_buchi}$ is quite delicate: we check not only that $\Sigma^0_2\ind$ implies Bounded-width K\"onig's Lemma, but also that constructing the objects to which we apply the lemma is within the means of $\rca + \Sigma^0_2\ind$.

The structure of the paper is as follows. Sections~\ref{sec:intro_rev_math} and~\ref{sec:intro_automata} discuss the necessary background on reverse mathematics, automata, and \mso. The sections that follow contain proofs of various implications which jointly establish the
equivalence of $\ref{it:induction}$--$\ref{it:dec_mso}$. 
We prove 
$\ref{it:induction}\to \ref{it:ramsey}$ in Section~\ref{sec:ind-to-ramsey}, 
$\ref{it:ramsey}\to \ref{it:compl_buchi}$  in Section~\ref{sec:ramsey-to-compl},
$\ref{it:compl_buchi}\to \ref{it:dec_mso}$ in Section~\ref{sec:compl-to-dec},
$\ref{it:dec_mso} \to \ref{it:induction}$ in Section~\ref{sec:dec-to-ind},
and $\refComplBuchiPrime \to \ref{it:induction}$ in Section~\ref{sec:two-compl}.

The next two sections contain some supplementary results. In Section~\ref{ap:ramsey-to-ind}, we give a direct proof of 
$\ref{it:induction}$  from $\ref{it:ramsey}$ and also from a related form of Ramsey's Theorem.
Then in Section~\ref{sec:ind-to-wilke} we show that the Wilke
extension principle, an algebraic statement that is also sometimes used to prove decidability of \mso,
is also equivalent to $\ref{it:induction}$--$\ref{it:dec_mso}$.

We then turn our attention to the determinisation theorem for word automata, i.e.,~$\ref{it:det_buchi}$.
Section~\ref{sec:ind-to-konig} provides a~combinatorial proof that $\Sigma^0_2\ind$ implies Bounded\=/width K\"onig's Lemma, 
which is then applied to prove $\ref{it:induction}\to \ref{it:det_buchi}$ in Section~\ref{sec:ind-to-det}. 
A~short Section~\ref{sec:conclusion} contains some concluding remarks and open problems.


\section{Background on reverse mathematics}
\label{sec:intro_rev_math}


\emph{Reverse mathematics}~\cite{simpson} is a framework for studying the strength of axioms needed to prove theorems of countable mathematics, that is, the part of mathematics concerned with objects that can be represented using no more than countably many bits of information. This encompasses the vast majority of the mathematics needed in computer science. 

The basic idea of reverse mathematics is to analyse mathematical theorems in terms of subsystems of a strong axiomatic theory known as second-order arithmetic. The two-sorted language of second-order arithmetic, $L_2$, contains \emph{first-order} variables $x, y, z, \ldots$ (or $i,j,k,\ldots$), intended to range over natural numbers, and \emph{second-order} variables $X,Y,Z,\ldots$, intended to range over sets of natural numbers. $L_2$ includes the usual arithmetic functions and relations ${+}, {\cdot}, {\le}, 0,1$ on the first-order sort, and the ${\in}$ relation which has one first-order and one second-order argument.
The intended model of $L_2$ is $(\omega, \mathcal{P}(\omega))$.

\subsection*{Notational convention.} From this point onwards, we will use the letter $\Nn$ to denote the natural numbers as formalised in second-order arithmetic, that is, the domain of the first-order sort. On the other hand, the symbol $\omega$ will stand for the concrete, or standard, natural numbers. For instance, given a theory $\mathrm{T}$ and a formula $\varphi(x)$, ``$\mathrm{T}$ proves $\varphi(n)$ for all $n \! \in \! \omega$'' will mean ``$\mathrm{T} \vdash \varphi(0), \mathrm{T} \vdash \varphi(1), \ldots$'', which does not imply $\mathrm{T} \vdash \forall x\!\in\!\Nn.\,\varphi(x)$.

\bigskip
\noindent\emph{Full second-order arithmetic}, $\mathsf{Z}_2$, has axioms of three types: (i) axioms stating that the first-order sort is the non-negative part of a discretely ordered ring; (ii) comprehension axioms, or sentences of the form \[\forall \bar Y \, \forall \bar y\, \exists X\, \forall x\, \big(x\in X \Leftrightarrow \varphi(x, \bar Y, \bar y)\big),\] where $\varphi$ is an arbitrary formula of $L_2$ not containing the variable $X$; (iii) the induction axiom, 
\[\forall X\, \big[0 \in X \land \forall x\, (x \in X \Rightarrow x+1 \in X)\Rightarrow \forall x.\, x\in X \big].\]

The language $L_2$ is very expressive: already in weak fragments of $\mathsf{Z}_2$, the first-order sort can be used to encode arbitrary finite objects and the second-order sort can encode even such objects as complete separable metric spaces, continuous functions between them, and Borel sets within them (cf.~\cite[Chapters~II.5, II.6, and~V.3]{simpson}). Moreover, the theory $\mathsf{Z}_2$ is powerful enough to prove almost all theorems from a typical undergraduate course that are expressible in $L_2$. In fact, the basic observation underlying reverse mathematics~\cite{simpson} is that many important theorems are equivalent to various fragments of $\mathsf{Z}_2$, where the equivalence is proved in some specific weaker fragment, referred to as the \emph{base theory}. 

\subsection*{Quantifier hierarchies.} Typical fragments of $\mathsf{Z}_2$ are defined in terms of well\=/known quantifier hierarchies whose definitions we now recall. A formula is $\Sigma^0_n$ if it has the form $\exists \bar{x}_1\, \forall\bar{x}_2\, \ldots\, \mathsf{Q}\bar{x}_n.\, \psi$, where the $\bar x_i$'s are blocks of first-order variables, the shape of $\mathsf{Q}$ depends on the parity of $n$, and $\psi$ is $\Delta^0_0$, i.e.~contains only bounded first\=/order quantifiers. A formula is $\Pi^0_n$ if it is the negation of a $\Sigma^0_n$ formula. A formula is \emph{arithmetical} if it contains only first-order quantifiers (second-order parameters are allowed). 

A formula is $\Sigma^1_n$ if it has the form $\exists \bar{X}_1\, \forall\bar{X}_2\, \ldots \mathsf{Q}\bar{X}_n.\, \psi$, where the $\bar X_i$'s are blocks of second-order variables, the shape of $\mathsf{Q}$ depends on the parity of $n$, and $\psi$ is arithmetical. A formula is $\Pi^1_n$ if it is the negation of a $\Sigma^1_n$ formula.

In practice, we say that a formula is $\Sigma^i_n/\Pi^i_n$ if it  equivalent to a $\Sigma^i_n/\Pi^i_n$ formula 
in the axiomatic theory we are working in at a given point.

\subsection*{Definition of $\rca$.} The usual base theory in reverse mathematics is $\rca$, which guarantees only the existence of decidable sets. $\rca$ is defined by restricting the comprehension scheme to \emph{$\Delta^0_1$}-comprehension, which takes the form:
\[\forall \bar Y \, \forall \bar y\, \left[\forall x\, (\varphi(x,\bar Y, \bar y) \Leftrightarrow \neg\psi(x,\bar Y, \bar y)) \Rightarrow \exists X\, \forall x\, (x\in X \Leftrightarrow \varphi(x, \bar Y, \bar y)) \right],\]
where both $\varphi$ and $\psi$ are $\Sigma^0_1$ and do not contain $X$. For technical reasons, it is necessary to strengthen the induction axiom to $\Sigma^0_1\ind$, that is, the axiom scheme consisting of the sentences
\[\forall \bar Y \, \forall \bar y\, \left[\varphi(0,\bar Y, \bar y) \land \forall x\,\left(\varphi(x,\bar Y, \bar y) \Rightarrow \varphi(x+1,\bar Y, \bar y)\right)
\Rightarrow \forall x.\, \varphi(x,\bar Y, \bar y) \right]\]
for $\varphi$ in $\Sigma^0_1$.  The scheme $\Sigma^0_1\ind$ makes it possible to define sequences by primitive recursion (cf.~\cite[Theorem~II.3.4]{simpson}): given some $x_0$ and a function $\fun{f}{\Nn}{\Nn}$, $\rca$ proves that there is a unique sequence $(x_i)_{i \in \Nn}$ such that $x_{i+1}=f(x_i)$ for each $i$.

$\rca$ has a unique minimal model in the sense of embeddability. This minimal model is $(\omega, \mathrm{Dec})$, where $\mathrm{Dec}$ is the family of decidable subsets of $\omega$.

\subsection*{The $\Sigma^0_n\ind$ scheme.}

In this paper we study an extension of $\rca$ obtained by strengthening the induction scheme to $\Sigma^0_2$ formulae. In general, 
for $n \! \in \! \omega$, the axiom scheme $\Sigma^0_n\ind$ is defined like $\Sigma^0_1\ind$, but with the induction formula $\varphi$ allowed to be 
$\Sigma^0_n$ rather than just $\Sigma^0_1$.  For each $n$, $\rca + \Sigma^0_n\ind$ is equivalent to $\rca + \Pi^0_n\ind$, where the latter is defined in the natural way, as well as to the least number principle for $\Sigma^0_n$ or $\Pi^0_n$ formulae (cf.~\cite[Chapter~II.3]{simpson}). 

Two important principles provable from $\Sigma^0_n\ind$ are \emph{$\Sigma^0_n$-collection}:
\[\forall \bar Z \, \forall \bar z\,
\big[\forall x \! \le \! t \,\exists y.\,  \varphi(x,y,\bar Z, \bar z)\big]
\Rightarrow \exists w\, \forall x \! \le \! t \, \exists y \! \le \! w.\, \varphi(x,y,\bar Z, \bar z),\]
for $\varphi$ in $\Sigma^0_n$, and \emph{bounded $\Sigma^0_n$-comprehension}:
\[\forall \bar Y \, \forall \bar y\, \forall w \, \exists X\, \forall x\, (x\in X \Leftrightarrow x \le w \land \varphi(x, \bar Y, \bar y)),\]
for $\varphi$ in $\Sigma^0_n$. The combination of the two yields \emph{strong $\Sigma^0_n$-collection}:
\[\forall \bar Z \, \forall \bar z\, \forall t\, \exists w\, \forall x \! \le \! t \,
\big[\exists y. \, \varphi(x,y,\bar Z, \bar z)\Rightarrow \exists y \! \le \! w.\,\varphi(x,y,\bar Z, \bar z) \big].\]

For each $n$, the theory $\rca + \Sigma^0_{n+1}\ind$ is strictly stronger than $\rca + \Sigma^0_{n}\ind$ (cf.\ e.g. \cite[Theorem IV.1.29]{hp}). 
However, note that the minimal model $(\omega, \mathrm{Dec})$ of $\rca$ satisfies $\rca + \Sigma^0_n\ind$ for all $n$, because an induction axiom is always true in a model with first-order universe $\omega$.

\subsection*{Additive Ramsey and Bounded-width K\"onig.}

Two prominent extensions of $\rca$ are related to weak forms of important nonconstructive set existence principles: K\"onig's Lemma and Ramsey's Theorem.

\emph{Weak K\"onig's Lemma} is the statement: ``for every $k$, every infinite tree contained in $\{0,1,\ldots,k\}^\ast$ has an infinite branch''. The theory obtained by adding this statement to $\rca$ is known as $\wkl$. This is the minimal theory supporting all sorts of ``compactness arguments'' in combinatorics, topology, analysis, and elsewhere (cf.~\cite[Chapter~IV]{simpson}).

The theory $\rt^2_2$ extends $\rca$ by an axiom expressing \emph{Ramsey's Theorem for pairs and two colours}\footnote{By $[X]^2$ we denote the set of unordered pairs of elements of $X$.}: ``for every $2$-colouring of $[\Nn]^2$ there exists an infinite homogeneous set''. $\rt^2_{<\infty}$ is defined similarly but allowing $k$-colourings for each $k \in \Nn$. 

Both $\rt^2_2$ and $\rt^2_{<\infty}$ are known to be incomparable with $\wkl$ in the sense of implication over $\rca$~\cite{hirst:thesis, liu}.
$\wkl$, $\rt^2_2$, and $\rt^2_{<\infty}$ are all false in the minimal model $(\omega, \mathrm{Dec})$ of $\rca$, see~\cite{jockusch:ramsey,kreisel:variant}. Much more on these theories can be found in~\cite{hirschfeldt}.

In this paper, we study specific restricted versions of $\rt^2_{<\infty}$ and $\wkl$  which play a role in proofs of B\"uchi's theorem. Recall that a \emph{semigroup} is a set $S$ with an associative operation $\fun{{*}}{S\times S}{S}$.

\begin{definition}[Additive Ramsey's Theorem]
\label{def:additive-ramsey}
\emph{Additive Ramsey's Theorem} is the following statement: for every finite semigroup $(S, {*})$ and every colouring $\fun{\colA}{[\Nn]^2}{S}$ such that for every $i \! < \! j \! < \! k$ we have $\colA(i,j) * \colA(j, k) = \colA(i, k)$, there exists an infinite homogeneous set $I \subseteq \Nn$. That is, there is a fixed colour $a$ such that for every $(i, j) \in [I]^2$, $\colA(i, j) = a$.
\end{definition}



\comment{
$\rca + \Sigma^0_2\ind$ is too weak to prove Weak K\"onig's Lemma (in fact, $\Sigma^0_2\ind$ and $\wkl$ are incomparable over $\rca$). However, it turns out that $\Sigma^0_2\ind$ proves a restricted version of the lemma, where the ``width'' of the trees under consideration is globally bounded, in the sense that the subtree rooted in 
a vertex $\tuple{i_0,\ldots,i_{\ell}} \in \{0,\ldots,k\}^*$ is completely determined by $\ell$ and $i_\ell$.
}

\begin{definition}[Bounded-width K\"onig's Lemma]
\label{def:bounded-konig}
\emph{Bounded-width K\"onig's Lemma} is the following statement: for every finite set $Q$ and every graph $G$ whose vertices belong to $Q\times\Nn$ and whose edges are all of the form \mbox{$((q,i),(q',i+1))$} for some $q,q'\in Q$, $i\in\Nn$, if there are arbitrarily long finite paths in $G$ starting in some vertex $(q,0)$, then there is an infinite path in $G$ starting in $(q,0)$.
\end{definition}

Notice that Bounded-width K\"onig's Lemma applied to a graph $G$ is essentially the same as Weak K\"onig's Lemma applied to the tree obtained by the so-called unraveling of $G$ (in particular, Bounded-width K\"onig's Lemma is provable in $\wkl$). This tree has globally bounded width in the sense that the number of vertices at each depth is always bounded by $|Q|$. However, we feel that the graph formulation is more natural to express. One of our results (Theorem~\ref{thm:ind-to-konig}) is that while induction is insufficient to prove Weak K\"onig's Lemma (in fact, $\Sigma^0_n\ind$ and $\wkl$ are incomparable over $\rca$ for all $n \ge 2$), the bounded-width variant follows from $\Sigma^0_2\ind$.

Some restrictions of Weak K\"onig's Lemma equivalent to the Bounded-width version have been independently studied in~\cite{sy:very-weak-wkl}.


\section{Background on \texorpdfstring{$\mso$}{MSO} and B\"uchi automata}
\label{sec:intro_automata}

B\"uchi automata and \mso logic are equivalent formalisms for specifying properties of infinite words. In this section we formally introduce these concepts. If not stated otherwise, the formalisation presented here is carried out in $\rca$.

\subsection*{Infinite words}

By $\Sigma$ we denote a~finite nonempty set  called an \emph{alphabet}. A \emph{finite word} over $\Sigma$ is a~function $\fun{w}{\{0,\ldots,k-1\}}{\Sigma}$; the \emph{length} of $w$ is $k$. The set of all finite words over $\Sigma$ is denoted by $\Sigma^\ast$ and the set of non-empty finite words is denoted by $\Sigma^{{+}}$. An \emph{infinite word} over $\Sigma$ is a~function $\fun{\infA}{\Nn}{\Sigma}$. We write $\infA \in \Sigma^\Nn$ for ``$\infA$ is an infinite word over $\Sigma$''.

Every infinite word can be treated as a~relational structure with universe $\Nn$, the binary order relation ${\leq}$, and a~unary relation $a$ for every $a\in \Sigma$. The semantics of these relations over a~given infinite word $\infA$ is natural, in particular $a(x)$ holds if $\infA(x)=a$.

When working with automata and logic it is natural to consider \emph{languages}---sets of infinite words satisfying certain properties. However, from the point of view of second-order arithmetic an infinite word is a~second-order object, so a~language would be a~``third-order object'' to which we do not have access. Therefore, in this paper we avoid talking directly about languages. Instead, when we want to express some properties of languages, we explicitly quantify over infinite words with a~given property.

\subsection*{Automata over infinite words}

A (nondeterministic) B\"uchi automaton is a~tuple $\Aa=\langle Q, \Sigma, q_\init, \delta, F\rangle$ where: $Q$ is a~finite set of \emph{states}, $\Sigma$ is an alphabet, $q_\init\in Q$ is an \emph{initial state}, $\delta\subseteq Q\times \Sigma\times Q$ is the \emph{transition relation}, and $F\subseteq Q$ is the set of \emph{accepting states}. Given an infinite word $\infA\in \Sigma^\Nn$, we say that $\rho\in Q^\Nn$ is a~\emph{run} of $\Aa$ over $\infA$ if $\rho(0)= q_\init$ and for every $n\in\Nn$ we have $\big(\rho(n),\infA(n),\rho(n+1)\big)\in\delta$. A run $\rho$ is \emph{accepting} if $\rho(n)\in F$ for infinitely many $n\in\Nn$. An automaton $\Aa$ \emph{accepts} $\infA$ if there exists an accepting run of $\Aa$ over $\infA$. An automaton is \emph{deterministic} if for every $q\in Q$ and $a\in\Sigma$ there is at most one transition of the form $(q,a,q') \in \delta$. When the automaton is not clear from the context, we put it in the superscript, i.e.~$Q^\Aa$ is the set of states of $\Aa$.

The possible transitions of a~B\"uchi automaton over a~particular letter $a\in\Sigma$ can be encoded as a~\emph{transition matrix} $\fun{M_a}{Q\times Q}{\{\trNon,\trYes,\trAcc\}}$, where $M_a(q,q')=\trNon$ if $(q,a,q')\notin\delta$, otherwise $M_a(q,q')=\trAcc$ if $q\in F$, and otherwise $M_a(q,q')=\trYes$. Let $[Q]$ be the set of all such functions $\fun{M}{Q\times Q}{\{\trNon,\trYes,\trAcc\}}$.

Since deterministic B\"uchi automata are strictly weaker than general B\"uchi automata \cite{PinPerrin}, one introduces the more flexible \emph{Rabin acceptance condition} in order to determinise B\"uchi automata. A \emph{Rabin automaton} is a~tuple $\Aa=\langle Q, \Sigma, q_\init, \delta, (E_i, F_i)_{i=1}^k \rangle$ as in the case of B\"uchi automata, where $E_i,F_i\subseteq Q$ for $i=1,\ldots,k$. A run $\rho \in Q^\Nn$ of $\Aa$ is \emph{accepting} if and only if for some $i\in\{1,\ldots,k\}$ each state in $E_i$ appears in $\rho$ only finitely many times and some state in $F_i$ appears in $\rho$ infinitely many times.

In general (i.e.~in $\mathsf{Z}_2$) Rabin automata can easily be complemented into so-called \emph{Streett automata}, and both classes can be transformed into nondeterministic B\"uchi automata. However, the transformations into B\"uchi automata require more than $\rca$. For Streett automata, $\Sigma^0_2\ind$ seems necessary. For Rabin, we need the B\"uchi automaton to guess that no state from a~specific set $E_i$ will reappear in the run under consideration. To prove that such a~construction is correct one needs $\Sigma^0_2$-collection---within $\rca$ the fact that in a~given run $\rho$ each state $q\in E_i$ appears only finitely many times does not imply a~global bound after which no state from $E_i$ reappears. That is the essential reason why it is not clear whether \ref{it:det_buchi} of Theorem~\ref{thm:main} implies the other items
in $\rca$.

\subsection*{Monadic Second-Order logic}

Monadic second-order logic (\mso) is an extension of first-order logic. 
\mso logic allows: boolean connectives ${\lnot}$, ${\lor}$, ${\land}$; the first-order quantifiers $\exists x$ and $\forall x$; and the monadic second-order quantifiers $\exists X$ and $\forall X$, where the variable $X$ ranges over subsets of the universe. Apart from predicates from the signature of a~given structure, the logic admits the binary predicate $x\in X$ with the natural semantics.

\subsection*{Definition of truth for $\mso$ over $\Nn$}

In order to state our theorems involving decidability of the $\mso$ theory 
of $(\Nn,{\le})$, we need to formulate the semantics of  monadic second-order logic within $\rca$. This involves a~coding of formulae $\phi \mapsto \godelize{\phi}$; we identify a~formula with its code. However, in second-order arithmetic there is no canonical definition of truth in an infinite structure which would work for all of \mso. Moreover, by Tarski's theorem on the undefinability of truth, for some infinite structures there is no such definition at all. In particular, it is not at all clear how to state the decidability of  $\mso(\Nn,{\le})$ as a~single sentence.

On the other hand, already $\rca$ is able to express a~truth definition for the \emph{depth\=/$n$ fragment of $\mso$}, for each $n \in \omega$. Here the depth of a~formula is calculated as the largest number of alternating blocks of $\land$/$\forall$'s and $\lor$/$\exists$'s appearing on a~branch in the syntactic tree of the formula (assume that all negations are pushed inside using the De Morgan laws). Essentially, the truth definition needs 
one universal set quantifier for a~block of $\land$/$\forall$'s and 
one existential set quantifier for a~block of $\lor$/$\exists$'s\footnote{After slight modifications, the truth definition would still work if we allowed depth-$n$ formulae to contain arbitrarily many alternations $\land$'s and $\lor$'s inside the scope of the last quantifier counted towards the depth.}.

So, what is possible is to provide formulae $\varphi_n$ stating that the depth-$n$ fragment of $\mso(\Nn,{\le})$ is decidable. We show in Section~\ref{sec:compl-to-dec} that every $\varphi_n$ can be proved in $\rca$ assuming a~complementation procedure for B\"uchi automata. On the other hand, 
we show in Sections~\ref{sec:ind-to-ramsey}, \ref{sec:ramsey-to-compl} that the existence of such a~procedure follows from $\Sigma^0_2\ind$, 
and in Section~\ref{sec:dec-to-ind} that $\varphi_5$ implies $\Sigma^0_2\ind$. As a~corollary, $\rca \vdash \varphi_5 \Rightarrow \varphi_n$ 
for every $n \in \omega$. In fact, assuming $\varphi_5$ there is a~single 
algorithm which provably in $\rca$ witnesses $\varphi_6, \varphi_7, \ldots$.

\subsection*{The B\"uchi decidability theorem}

In~\cite{buchi1962} B\"uchi proved decidability of the theory $\mso(\Nn,{\le})$. In the context of B\"uchi's procedure, it is natural to evaluate $\mso$ sentences not just on $\mso(\Nn,{\le})$, but also on infinite words over various alphabets $\Sigma$. These infinite words are represented as expansions of $(\Nn,{\le})$ by unary relations, in the way explained above.

The following theorem captures as much of B\"uchi's result as can be naturally expressed in relatively weak theories of second-order arithmetic.

\begin{theorem}[B\"uchi formalised] 
\label{thm:buchi-in-rca}
There exists an effective procedure $p$ such that for every fixed depth $n\!\in \!\omega$ the following is provable in $\rca+\Sigma^0_2\ind$. For every statement $\phi$ of \mso over an alphabet $\Sigma$ such that the depth of $\phi$ is at most $n$, the procedure $p(\phi)$ produces a~nondeterministic B\"uchi automaton $\Aa$ over $\Sigma$ such that for every infinite word $\infA\!\in\!\Sigma^\Nn$, this word satisfies $\phi$ if and only if $\Aa$ accepts $\infA$. Moreover, it is decidable if a~given nondeterministic B\"uchi automaton accepts at least one infinite word.
\end{theorem}

We discuss some issues related to formalising the inductive proof of B\"uchi's theorem in Section~\ref{sec:compl-to-dec}. The crucial step concerns complementation of automata, which is used to treat negations of subformulae in $\phi$ (or subformulae beginning with $\forall$, assuming the negations have been pushed inside).


\section{\texorpdfstring{$\Sigma^0_2\ind$}{IND} implies Additive Ramsey}
\label{sec:ind-to-ramsey}

The aim of this section is to prove the following proposition, which is implication $\ref{it:induction} \to \ref{it:ramsey}$ of Theorem \ref{thm:main}.

\newcommand{\thmIndToRamsey}{Over $\rca$, $\Sigma^0_2\ind$ implies Additive Ramsey's Theorem.}

\begin{proposition}
\label{thm:ind-to-ramsey}
\thmIndToRamsey
\end{proposition}

The proof of Proposition \ref{thm:ind-to-ramsey} consists of two steps. First, we prove another weakening of Ramsey's Theorem.

\begin{definition}
\label{def:ordered-ramsey}
\emph{Ordered Ramsey's Theorem} for pairs states that if $(P, \preceq)$ is a finite partial order and $\fun{\colA}{[\Nn]^2}{P}$ is a colouring such that for every $i\! <\! j\! <\! k$ we have $\colA(i, j) \succeq \colA(i, k)$, then there exists an infinite homogeneous set $I\subseteq \Nn$, i.e.~$\colA(i,j)=\colA(i',j')$ for all $(i,j),(i',j')\in [I]^2$.
\end{definition}

It will follow from Lemma~\ref{lem:orderedramsey} below and the proof of Proposition~\ref{lem:ramsey-to-ind} in Section~\ref{ap:ramsey-to-ind} that Ordered Ramsey's Theorem is equivalent to its restriction to linear orders, and thus to the case where $P$ is $\{0,\ldots,n\}$ for some $n \in \Nn$ and ${\preceq}$ is the usual ordering. Note also that the theorem follows immediately from the so-called  \emph{Stable Ramsey's Theorem} $\srt^2_{<\infty}$ (cf.~\cite[Sections~6.4 and~6.8]{hirschfeldt}), where the requirement on $\colA$ is only that $\colA(i,\cdot)$ should stabilise for each $i$.

\begin{lemma}
\label{lem:orderedramsey}
Over $\rca$, $\Sigma^0_2\ind$ proves Ordered Ramsey's Theorem.
\end{lemma}

\begin{proof}
We call a colour $p\in P$ \emph{recurring} if $\forall i\, \exists k\!>\!j\!>\!i.\ \colA(j,k)=p$. Notice that for each non-recurring colour $p$ there exists $i_p$ such that there is no occurrence of $p$ to the right of $i_p$ (i.e.~no $k>j>i_p$ such that $\colA(j,k)=p$). By an application of strong $\Sigma^0_2$-collection we obtain some $i_0$ such that for every non-recurring colour $p$ and every $k>j>i_0$ we have $\colA(j,k)\neq p$. In particular, there is a recurring colour. Moreover, being a recurring colour is a $\Pi^0_2$ property, so by $\Sigma^0_2\ind$ we can find a ${\preceq}$-minimal recurring colour $p_0$. 

We now define a sequence $(u_i,v_i)_{i \in \Nn}$ by primitive recursion on $i$. Let $(u_0, v_0)$ be some pair such that $i_0 < u_0 < v_0$ and $c(u_0, v_0) = p_0$. Now assume that $u_0 < v_0 \le u_1 < v_1  \ldots \le u_i < v_i$ have been defined, $\{u_0,\ldots,u_i\}$ is homogeneous with colour $p_0$, and $\colA(u_i,v_i)= p_0$. Let $(u_{i+1}, v_{i+1})$ be the smallest pair such $v_i \le u_{i+1} < v_{i+1}$ and $\colA(u_{i+1}, v_{i+1}) = p_0$. Such a pair exists because $p_0$ is recurring. We know that $\colA(u_i,u_{i+1}) = p_0$, since on the one hand $\colA(u_i,u_{i+1}) \preceq \colA(u_i,v_i) = p_0$, and on the other hand $u_i > i_0$ and thus  $\colA(u_i,u_{i+1})$ is a recurring colour, so it cannot be $\preceq$-strictly smaller than $p_0$. Similarly, for $j<i$ we know that $\colA(u_j,u_{i+1}) = p_0$ because $\colA(u_j,u_{i+1}) \preceq p_0$ and $u_j > i_0$. Therefore, the set $\{u_i\mid i \in \Nn\}$ is homogeneous for $\colA$.
\end{proof}

Before proceeding to prove the additive version of Ramsey's Theorem, we recall a few basic facts about finite semigroups we shall use in our proof. The facts are proved by elementary combinatorial arguments which readily formalise in $\rca$. The proofs can be found for instance in~\cite{PinPerrin}.

\begin{definition}
\label{def:green}
Green preorders over a semigroup $S$ are defined as follows
\begin{itemize}
\item $s \le_{\GrR} t$ if and only if $s=t$ or $s \in t *S=\{t * a\mid a\in S\}$,
\item $s \le_{\GrL} t$ if and only if $s=t$ or $s \in S *t=\{a * t\mid a\in S\}$,
\item $s \le_{\GrH} t$ if and only if $s \le_{\GrR} t$ and $s \le_{\GrL} t$,
\item ${\le_{\GrJ}}$ is the transitive closure of the union of ${\le_{\GrR}}$ and ${\le_{\GrL}}$.
\end{itemize}
The associated equivalence relations are written $\GrR$, $\GrL$, $\GrH$, $\GrJ$; their equivalence classes are called respectively $\GrR$, $\GrL$, $\GrH$, and $\GrJ$-classes.
\end{definition}

\begin{lemma}
\label{lem:leLRH}
For every finite semigroup $S$ and $s, t \in S$, $s \le_{\GrL} t$ and $s \GrR t$ implies $s \GrH t$.
\end{lemma}

\begin{lemC}[{\cite[Proposition~2.4]{PinPerrin}}]
\label{lem:Hidemgroup}
If $(S, {*})$ is a finite semigroup, $H\subseteq S$ an $\GrH$-class, and some $a,b\in H$ satisfy $a*b\in H$ then for some $e\in H$ we know that $(H, {*}, e)$ is a group.
\end{lemC}

Now we can prove the main result of the section.

\begin{proofof}{Proposition~\ref{thm:ind-to-ramsey}}
Let a colouring $\colA$ take values in the finite semigroup $(S,{*})$ and satisfy the additivity condition of Definition~\ref{def:additive-ramsey}. For every position $i$ and every $k \geq j > i$, let us observe that $\colA(i,k) \le_{\GrR} \colA(i, j)$. Let $r$ be the function mapping every element of $S$ to its $\GrR$-class. The function $r \circ \colA$ is an ordered colouring with respect to~${\le_{\GrR}}$; let us use Lemma~\ref{lem:orderedramsey} to obtain a homogeneous sequence $(u_i)_{i \in \Nn}$ for $r \circ \colA$.

Since $S$ is finite, we can use $\Sigma^0_2$-collection to prove that there is some colour $a$ such that $\colA(u_0,u_{i}) = a$ for infinitely many $i$. This allows us to take a subsequence $(v_i)_{i \ge 0}$ of $(u_i)_{i \ge 0}$ such that $\colA(v_0,v_{i})=a$ for each $i$.

We now know that $a = a * \colA(v_i, v_j)$ for every $0 < i < j$. In particular, $a \le_{\GrL} \colA(v_i, v_j)$ by the definition of ${\le_{\GrL}}$. Since $a$ and $\colA(v_i, v_j)$ are $\GrR$-equivalent, Lemma~\ref{lem:leLRH} implies that $\colA(v_i, v_j) \GrH a$. Let $H$ be the $\GrH$-class of $a$. Since $a*\colA(v_i,v_j)=a\in H$, we know by Lemma~\ref{lem:Hidemgroup} that $(H,{*},e)$ is a group for some $e\in H$. Using this group structure and the equation $a = a * \colA(v_i, v_j)$ we obtain that $\colA(v_i, v_j)=e$. Hence, $\{v_{i+1} \mid i \in \Nn \}$ is a homogeneous set for $\colA$ with the colour $e$.
\end{proofof}


\section{Additive Ramsey implies complementation}
\label{sec:ramsey-to-compl}

In this section, we sketch a proof of the following result, which is implication $\ref{it:ramsey}\to \ref{it:compl_buchi}$ of Theorem \ref{thm:main}.

\begin{proposition}
\label{thm:ramsey-to-compl}
Over $\rca$, Additive Ramsey's Theorem proves the correctness of the standard complementation procedure for 
B\"uchi automata: given a B\"uchi automaton $\Aa$ over an alphabet $\Sigma$, the procedure outputs a B\"uchi automaton $\Bb$ over the same alphabet such that for every $\infA\in \Sigma^\Nn$ we have that $\Aa$ accepts $\infA$ if and only if $\Bb$ does not accept $\infA$.
%
\end{proposition}

The proof of this result follows the usual construction of the automaton $\Bb$~\cite{buchi1962}: the states of $\Bb$ are based on transition matrices of $\Aa$ (see Section~\ref{sec:intro_automata}). The automaton $\Bb$ guesses a Ramseyan decomposition of the given infinite word $\infA$ with respect to a certain homomorphism into $[Q]$; and then verifies that the decomposition witnesses that there cannot be any accepting run of $\Aa$ over $\infA$.

Let us fix a B\"uchi automaton $\Aa=\langle Q, \Sigma, q_\init, \delta, F\rangle$. We will introduce a semigroup structure on the set of all transition matrices of $\Aa$. Let us define the natural operations of addition and multiplication over $\{\trNon,\trYes,\trAcc\}$ as depicted on Figure~\ref{fig:operations}. The addition makes it possible to choose a preferred run (i.e.~an accepting transition is better than a non-accepting one) and the multiplication corresponds to concatenation of runs.

\begin{figure}
\centering
\begin{tikzpicture}[scale=0.5, every node/.style={scale=0.8}]
\coordinate (sero) at (-0.5,+0.5);

\foreach \z in {0,1,2,3} {
	\draw ($(sero)+(0.5,-0.5)+(-1,-\z)$) -- ($(sero)+(0.5,-0.5)+(3,-\z)$);
	\draw ($(sero)+(0.5,-0.5)+(\z,+1)$) -- ($(sero)+(0.5,-0.5)+(\z, -3)$);
}

\node at ($(sero) + (+0,-0)$) {${+}$};
\node at ($(sero) + (+1,-0)$) {$\trNon$};
\node at ($(sero) + (+2,-0)$) {$\trYes$};
\node at ($(sero) + (+3,-0)$) {$\trAcc$};

\node at ($(sero) + (+0,-1)$) {$\trNon$};
\node at ($(sero) + (+1,-1)$) {$\trNon$};
\node at ($(sero) + (+2,-1)$) {$\trYes$};
\node at ($(sero) + (+3,-1)$) {$\trAcc$};

\node at ($(sero) + (+0,-2)$) {$\trYes$};
\node at ($(sero) + (+1,-2)$) {$\trYes$};
\node at ($(sero) + (+2,-2)$) {$\trYes$};
\node at ($(sero) + (+3,-2)$) {$\trAcc$};

\node at ($(sero) + (+0,-3)$) {$\trAcc$};
\node at ($(sero) + (+1,-3)$) {$\trAcc$};
\node at ($(sero) + (+2,-3)$) {$\trAcc$};
\node at ($(sero) + (+3,-3)$) {$\trAcc$};

\coordinate (sero) at (+6.5,+0.5);

\foreach \z in {0,1,2,3} {
	\draw ($(sero)+(0.5,-0.5)+(-1,-\z)$) -- ($(sero)+(0.5,-0.5)+(3,-\z)$);
	\draw ($(sero)+(0.5,-0.5)+(\z,+1)$) -- ($(sero)+(0.5,-0.5)+(\z, -3)$);
}

\node at ($(sero) + (+0,-0)$) {${*}$};
\node at ($(sero) + (+1,-0)$) {$\trNon$};
\node at ($(sero) + (+2,-0)$) {$\trYes$};
\node at ($(sero) + (+3,-0)$) {$\trAcc$};

\node at ($(sero) + (+0,-1)$) {$\trNon$};
\node at ($(sero) + (+1,-1)$) {$\trNon$};
\node at ($(sero) + (+2,-1)$) {$\trNon$};
\node at ($(sero) + (+3,-1)$) {$\trNon$};

\node at ($(sero) + (+0,-2)$) {$\trYes$};
\node at ($(sero) + (+1,-2)$) {$\trNon$};
\node at ($(sero) + (+2,-2)$) {$\trYes$};
\node at ($(sero) + (+3,-2)$) {$\trAcc$};

\node at ($(sero) + (+0,-3)$) {$\trAcc$};
\node at ($(sero) + (+1,-3)$) {$\trNon$};
\node at ($(sero) + (+2,-3)$) {$\trAcc$};
\node at ($(sero) + (+3,-3)$) {$\trAcc$};
\end{tikzpicture}
\caption{Two operations on $\{\trNon, \trYes,\trAcc\}$ used to define multiplication on $[Q]$.}
\label{fig:operations}
\end{figure}

Now, given two transition matrices $M,N\in [Q]$ we can naturally define the matrix $M * N$ that is obtained by the standard matrix multiplication formula. 
Notice that the mapping $\Sigma\ni a\mapsto M_a\in[Q]$ from Section~\ref{sec:intro_automata} can be extended to a homomorphism $\fun{\homA}{\Sigma^\ast}{[Q]}$. Clearly, for a finite word $u\in\Sigma^\ast$ 
the matrix $\homA(u)$ represents possible runs of $\Aa$ over $u$, in analogy to the way in which $M_a$ represents
possible transitions over $a$.

We will say that a pair $(N,M)\in [Q]\times[Q]$ is \emph{rejecting} if:
\begin{itemize}
\item $N * M = N$,
\item $M * M= M$,
\item but there is no $q\in  Q$ such that $N(q_\init,q)\in\{\trYes,\trAcc\}$ and $M(q,q)=\trAcc$.
\end{itemize}

The structure of the automaton $\Bb$ is as follows: its set of states is $([Q])^3\cup ([Q])^2\cup [Q]\cup\{q_\init\}$. Intuitively, the automaton needs to guess that a given infinite word admits a homogeneous decomposition where the initial fragment has type $N$ and the homogeneous colour is $M$, for a rejecting pair $(N,M)$. The initial state of the automaton is $q_\init$. The accepting states are those in $[Q]$. The automaton has the following transitions (we write $K\trans{a}K'$ for a transition $(K,a,K')\in\delta$):

\begin{itemize}
\item $q_\init\trans{a} (N,M,M_a)$ for all rejecting pairs $(N,M)$,
\item $(N,M,K)\trans{a} (N,M,K * M_a)$,
\item $(N,M,K)\trans{a} M$, if $K * M_a=N$,
\item $M\trans{a} (M,M_a)$,
\item $M\trans{a} M$ if $M_a=M$,
\item $(M,K)\trans{a} (M,K * M_a)$,
\item $(M,K)\trans{a} M$, if $K * M_a=M$.
\end{itemize}

To complete the proof of Proposition \ref{thm:ramsey-to-compl}, it remains to show the following.

\begin{lemma}
\label{sublem:ramsey-to-compl}
Over $\rca$, Additive Ramsey's Theorem implies that for every infinite word $\infA$ the automaton $\Bb$ described above accepts $\infA$ if and only if the automaton $\Aa$ does not accept~$\infA$.
\end{lemma}

\begin{proof}
First assume that both $\Aa$ and $\Bb$ accept an infinite word $\infA$. Let $\rho$ be an accepting run of $\Aa$ and let $\tau$ be an accepting run of $\Bb$. Let the state $\tau(1)$ be $(N,M,K)$. Since $\tau$ is accepting, we know that $\tau$ visits a state from $[Q]$ infinitely many times.

The only possible such state is $M$. Taking $k_0 < k_1 < \ldots$ such that $\tau (k_i ) = M$ for each $i$, we can
decompose $\infA$ as $\infA=u_0u_1\ldots$ where the length of $u_0 u_1 . . . u_i$ is $k_i$. Then $\homA(u_0)= N$ and $\homA(u_i)=M$ for all $i>0$. 
Our aim is to find a state $q$ such that for some $j>i>0$ we have $\rho(k_i)=\rho(k_j)=q$ and there is some $\ell$ such that $k_i\leq \ell < k_j$ and $\rho(\ell)\in F$. We can find such $q$ using the pigeonhole principle: first define $\ell_0=1$ and then let $\ell_{i+1}$ be the smallest number such that there is an accepting state in $\rho$ between $k_{\ell_{i}}$ and $k_{\ell_{i+1}}$. The sequence $(\ell_i)_{i \in \Nn}$ is defined by primitive recursion, therefore it can be constructed in $\rca$. By the (finite) pigeonhole principle, there exist $0\leq i < j\leq |Q|+1$ such that $\rho(k_{\ell_i})=\rho(k_{\ell_j})=q$. Since $M*M=M$ and $\rho$ has an accepting state between $k_{\ell_i}$ and $k_{\ell_j}$ we know that $M(q,q)=\trAcc$. Similarly, since $N * M=N$, we know that $N(q_\init,q)\in\{\trYes,\trAcc\}$. It means that the pair $(N,M)$ is not rejecting, which contradicts the definition of the transitions of~$\Bb$.

Now assume that the automaton $\Bb$ rejects a given infinite word $\infA$. Consider a colouring $\colA$ such that for $i<j$ we have $\colA(i,j)=\homA\big(\infA(i)\infA(i+1)\cdots\infA(j-1)\big)$. Since $\homA$ is a homomorphism, we know that $\colA$ is additive. By Additive Ramsey's Theorem,
we can find $k_0 < k_1 < \ldots$ forming a homogeneous set for $C$. Decomposing $\infA=u_0u_1\ldots$ with $k_i$ the length of $u_0u_1\ldots u_i$ as previously, we have some $N,M\in[Q]$ such that $M * M = M$, $\homA(u_0)= N$ and $\homA(u_i)=M$ for all $i>0$.
by skipping the first element of the homogeneous set. If the pair $(N,M)$ was rejecting, the automaton $\Bb$ would accept $\infA$---we would be able to define using $\Delta^0_1$-comprehension an accepting run $\tau$ of $\Bb$ over $\infA$ such that $\tau(k_i)=M$ for all $i>1$. Therefore, there exists a state $q$ of the kind disallowed by the definition of a rejecting pair. This state can be used to construct an accepting run $\rho$ of $\Aa$ over $\infA$, such that for every $i>0$ we have $\rho(k_i)=q$. As above, such a run can be defined by $\Delta^0_1$\=/comprehension.
\end{proof}


\section{Effective complementation implies decidability}
\label{sec:compl-to-dec}

The following gives implication $\ref{it:compl_buchi}\to \ref{it:dec_mso}$ of Theorem \ref{thm:main}.

\begin{proposition}
\label{thm:compl-to-dec}
For each $n\in\omega$, $\rca$ proves: if there exists an algorithm for complementing B\"uchi automata, then there exists an algorithm which, given an \mso formula $\varphi$ of depth at most $n$, outputs an automaton $\Aa_\varphi$ such that for every word $\infA$, the formula $\varphi$ is satisfied by $\infA$ if and only if $\Aa_\varphi$ accepts $\infA$. As a consequence, the depth-$n$ fragment of $\mso(\Nn,{\le})$ is decidable.
\end{proposition}

\begin{remark}
In fact, the algorithm producing $\Aa_\varphi$ on input $\varphi$ is the same for each $n$. This is because there is a standard procedure
(in the terminology of computability theory, a \emph{Turing functional}) for converting algorithms for complementing B\"uchi automata into
algorithms deciding $\mso(\Nn, {\le})$. The proof of Proposition~\ref{thm:compl-to-dec} verifies that the algorithm obtained by this procedure
is, provably in $\rca$, correct on depth-$n$ sentences for each $n \in \omega$.
\end{remark}

The proof of Proposition \ref{thm:compl-to-dec} is based on the usual idea: given $\varphi$, inductively construct automata $\Aa_\psi$ corresponding to increasingly complicated subformulae $\psi$ of $\varphi$. However, the formula ``$\Aa_\psi$ is equivalent to $\psi$'' as written is not $\Sigma^0_1$ (not even arithmetical, as it quantifies over infinite words), so induction for it is not available in $\rca$. To deal with that, we make sure that for $\varphi$ of depth $n$ the algorithm only makes $O(n)$ \emph{big steps}, with a single \emph{big step} corresponding to an entire block of quantifiers/connectives at a given depth within $\varphi$. In that way, we can reason by induction of fixed length $n$, which is available in $\rca$ for formulae of arbitrary complexity. 




\begin{proof}
We first note that w.l.o.g.\ we can restrict attention to depth-$n$ MSO formulae of the form $\psi$ or $\xi$ given by the following grammar:

\[
\begin{array}{lcl}
\psi &:=& \forall \bar X.\, \bigwedge_{i = 1}^k \xi_i \bnfalt A \bnfalt \neg A \\
\xi &:=& \exists \bar X.\, \bigvee_{i = 1}^k \psi_i \bnfalt A \bnfalt \neg A \\
A &:=& \Sing(X) \bnfalt \min X \le \min Y \bnfalt X\subseteq Y
\end{array}
\]
where $\Sing(X)$ means ``$X$ is a singleton'' and $\min(X) \le \min(Y)$ means ``either $Y$ is empty or there is an element of $X$ less than or equal to the smallest element of $Y$''. The reason is that provably in $\rca$, it is possible to perform the following operations on an $\mso$ formula:
\begin{itemize}
\item replace each first-order variable $x$ by a corresponding second-order variable $X$; translate $x \le y$ to $\min(X) \le \min(Y)$ and relativise quantifiers over $X$ to $\Sing$,
\item push negations downwards to the level of atomic formulae,
\item rearrange $\lor$'s and $\exists$'s (respectively, $\land$'s and $\forall$'s) lying at the same depth,
\end{itemize}
and obtain a formula of the same depth which is equivalent to the original one modulo the obvious identification of $x$ with $\{x\}$.
The benefit of doing so is that we obtain formulae containing solely second-order variables. We can then treat an assignment to the variables
$X_1,\ldots,X_k$ as an infinite word over the alphabet $\{0,1\}^k$.


We also note that given an automaton $\Aa$ over $\{0,1\}^k$, it is easy to construct an automaton over $\{0,1\}^{k+\ell}$ which behaves just like $\Aa$ and ignores the additional $\ell$ coordinates. For this reason, when describing the automaton~$\Aa_\psi$ assigned to a formula $\psi$, we can safely assume that the alphabet of $\Aa_\psi$ has exactly as many coordinates as there are free variables in $\psi$; in the later steps of the construction, extra coordinates can be added as needed.  

The algorithm assigning automata to subformulae of $\varphi$ works inductively as follows:

\begin{figure}
\begin{center}
\begin{tikzpicture}[node distance=3.5cm,on grid,auto] 
   \node[state,initial] (q_0)   {}; 
   \node[state,accepting] (q_1) [right=of q_0] {};
    \path[->] 
    (q_0) edge [trans, loop above] node {0} ()
          edge [trans] node {1} (q_1)
    (q_1) edge [trans, loop above] node {0} ();
    
\begin{scope}[shift={(0,-4)}] 
   \node[state,initial, accepting] (q_0)   {}; 
   \node[state,accepting](q_1) [right=of q_0] {};
    \path[->] 
    (q_0) edge [trans, loop above] node {$\begin{pmatrix} 0\\ 0\end{pmatrix}$} ()
          edge [trans] node {$\begin{pmatrix} 1\\ 0\end{pmatrix}\!,   \begin{pmatrix} 1\\ 1\end{pmatrix}$} (q_1)
    (q_1) edge [trans, loop above] node {$\{0,1\}^2$} ();
\end{scope}
    
\begin{scope}[shift={(0,-8)}] 
   \node[state,initial, accepting] (q_0)   {}; 
    \path[->] 
    (q_0) edge [trans, loop above] node {$\begin{pmatrix} 0\\ 0\end{pmatrix}$, $\begin{pmatrix} 0\\ 1\end{pmatrix}$, $\begin{pmatrix} 1\\ 1\end{pmatrix}$} ();
\end{scope}
\end{tikzpicture}
\end{center}
\caption{The automata $\Aa_\Sing$, $\Aa_{\min}$, and $\Aa_{{\subseteq}}$. The initial states of the automata are indicated by incoming arrows. The accepting states are marked by double circles. The transitions are represented by arrows, labelled by the respective letters.}
\label{fig:autsingmin}
\end{figure}
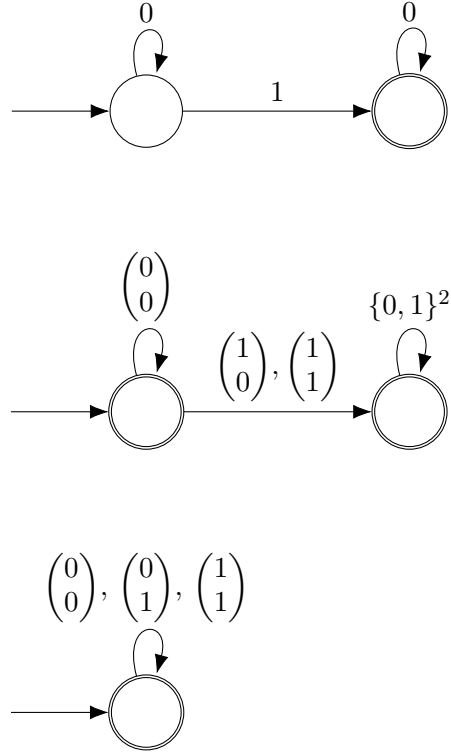

\begin{enumerate}

\item  The base case is for atomic subformulae of the form $\Sing(X)$, $\min(X) \le \min(Y)$, and $X\subseteq Y$. To these, the algorithm assigns the automata 
$\Aa_{\Sing}$, $\Aa_{\min}$, and $\Aa_{{\subseteq}}$, respectively, pictured in Figure~\ref{fig:autsingmin}. It is straightforward to verify in $\rca$ that the only situation in
which $\Aa_{\Sing}$ accepts a word over $\{0,1\}$ is if it encounters a~single position labelled $1$, switches to the accepting state, and remains in that state
by reading an infinite string of $0$'s. This happens exactly if the word represents a singleton set $X$. 
Similarly, it is straightforward to verify that a word over $\{0,1\}^2$ representing two sets $X, Y$ is accepted by $\Aa_{\min}$ (resp. $\Aa_{{\subseteq}}$) exactly if it is not the case
that $1$ appears on the second coordinate before it appears on the first coordinate (resp.~that $1$ appears on the first coordinate with $0$ on the second). This is just what is needed to recognise the property $Y{=}\emptyset\lor \min(X) {\le} \min(Y)$ (resp.~the property $X\subseteq Y$).


\item Automata corresponding to $\neg \Sing(X)$, $\neg \min(X) \le \min(Y)$, and $\neg X\subseteq Y$ are constructed using the algorithm for complementation.

\item Given formulae $\psi_i$, $1 \le i \le k$, and corresponding automata $\Aa_i = \langle Q_i , \{0,1\}^\ell , q_{\init_i} , \delta_i , F_i \rangle$, 
the automaton corresponding to $\bigvee_{1\le i \le k} \psi_i$ is $\bigvee_i \Aa_i := \langle \{ q_0 \} \sqcup \bigsqcup_i Q_i, \{0,1\}^\ell, q_0, \delta' \sqcup \bigsqcup_i \delta_i, \bigsqcup F_i  \rangle$ where the set $\delta'$ is $\{(q_0, a, q) \mid \exists i\!\le\! k.\ (q_{\init_i}, a, q) \in Q_i \}$.
If a word $\infA$ is accepted by some $\Aa_i$ due to a run $\rho$ then $\rho'$ defined as $\rho'(0)=q_0$ and as $\rho$ everywhere else is an accepting run of $\bigvee_i \Aa_i$ over $\infA$. Conversely, if $\rho \in (\{ q_0 \} \sqcup\bigsqcup_i Q_i)^\Nn$ is an accepting run of $\bigvee_i \Aa_i$ over $\infA$, 
then $\rho(1)$ belongs to $Q_i$ for some $i$. Then for $j>0$ each $\rho(j)$ also belongs to $Q_i$ and all corresponding transitions agree with 
$\delta_i$. Defining $\rho'$ by $\rho'(0) = q_{\init_i}$ and as $\rho$ everywhere else yields an accepting run of $\Aa_i$ over $\infA$.

\item If the automaton $\Aa = \langle Q , \{0,1\}^{k + \ell} , q_\init , \delta , F \rangle$ corresponding to $\psi(\bar X, \bar Y)$, 
then the automaton corresponding to $\exists \bar Y  \psi$ is $\exists \Aa := \langle Q, \{0,1\}^k, q_\init , \delta_\exists, F \rangle$ with 
$\delta_\exists := \{ (q, (a_1, \dots, a_k) , q') \mid \exists \bar b.\, (q, (a_1, \dots, a_k, b_1, \dots, b_\ell), q') \in \delta \}$. We argue that 
$\exists \Aa$ accepts a word $\infA$ if and only if there exists some $\infB \in (\{0,1\}^\ell)^\Nn$ such that $\Aa$ accepts $\infA \otimes \infB$,
where $\otimes$ stands for (coordinate-wise) concatenation of finite sequences.
Indeed, suppose that $\infA$ is accepted by $\exists \Aa$ using an accepting run $\rho \in Q^\Nn$. 
By the definition, this means that for every $j$ there exists $\bar b \in \{0,1\}^\ell$ such that $(\rho(j), \infA(j) \otimes \bar b, \rho(j + 1)) \in \delta$. 
Use $\Delta^0_1$-comprehension to define an infinite word $\infB$ by picking a minimal such $\bar b$ as $\infB(j)$ for every $j$. Then $\rho$ is an accepting  run of $\Aa$ over $\infA \otimes \infB$. 
Conversely, it is clear that an accepting run $\rho$ of $\Aa$ over $\infA \otimes \infB$ is an accepting run of $\exists \Aa$ over $\infA$.

\item Finally, the formula $\forall \bar X.\, \bigwedge_{i=1}^k \xi_i$ is equivalent to $\neg \exists \bar X.\, \bigvee_{i=1}^k \neg \varphi_i$. 
The automaton corresponding to it is built by means of constructions 3 and 4 and two rounds of complementations
\end{enumerate}
Clearly, we can argue by induction on $m \le n$ that for all subformulae $\psi$ of $\varphi$ at depth $m$, the automaton $\Aa_\psi$ is equivalent to $\psi$. 
In particular, $\Aa_\varphi$ is equivalent to $\varphi$.

It remains to deduce decidability of the depth-$n$ fragment of $\mso(\Nn,{\le})$. 
Given an algorithm transforming a depth-$n$ $\mso$ formula to an equivalent automaton, it suffices to 
show decidability of the emptiness problem for B\"uchi automata: 
``given a nondeterministic B\"uchi automaton $\Aa$, does there exist an infinite word accepted by $\Aa$?''
As is well known, the answer is positive exactly if $\Aa$ contains a state $q$ which is reachable from
the initial state $q_I$ and has the property that $q$ can be reached from $q$ via a path containing an accepting state.
The standard argument proving this formalises in $\rca$ in an unproblematic way.
\end{proof}

\section{Decidability implies \texorpdfstring{$\Sigma^0_2\ind$}{IND}}
\label{sec:dec-to-ind}

In this section we prove the following result.

\newcommand{\thmDecToInd}{Over $\rca$, the decidability of the depth\=/$5$ fragment of the theory $\mso(\Nn, {\le})$ implies $\Sigma^0_2\ind$.}

\begin{proposition}
\label{thm:dec-to-ind}
\thmDecToInd
\end{proposition}

This is, of course, implication $\ref{it:dec_mso} \to \ref{it:induction}$ of Theorem~\ref{thm:main}. The proof of the implication is based on two
observations which deserve to be stated as separate lemmas. 

The first lemma explains one way in which the decidability of the MSO theory of some structure can be used to derive some nontrivial principles. 
Basically, properties corresponding to families of MSO sentences are decidable (in particular, $\Sigma^0_1$), and therefore mathematical induction can be applied to them.

\begin{lemma}\label{lem:induction-made-simple}
For every $n \in \omega$, the following is provable in $\rca$. Let $(\psi_i)_{i \in \Nn}$ be a sequence of depth\=/$n$ MSO sentences and let 
$\mathbb{A}$ be a structure such that the depth\=/$n$ fragment of the theory $\mso(\mathbb{A})$ is decidable. If $\psi_0 \in \mso(\mathbb{A})$
and if $\psi_i \in \mso(\mathbb{A})$ implies $\psi_{i+1} \in \mso(\mathbb{A})$ for each $i \in \Nn$, then $\psi_i \in \mso(\mathbb{A})$ for each $i \in \Nn$. 
\end{lemma}

\begin{proof}
It is enough to note that the property ``$\psi_i \in \mso(\mathbb{A})$'' can be expressed by a $\Sigma^0_1$ $L_2$-formula $\varphi(i)$ 
(and, in fact, by a $\Pi^0_1$ formula too), and $\Sigma^0_1\ind$ is available.
\end{proof}

The second lemma will provide us with a concrete MSO\=/expressible property to which the first lemma can be applied.

\begin{lemma}\label{lem:pi-2-complete-word}
Let $\pi(i)$ be the $\Pi^0_2$ formula $\forall x \, \exists y. \,\delta(i,x,y)$, where $\delta(i,x,y)$ is~$\Delta^0_0$, possibly with parameters. Then $\rca$ proves that for every $k \in \Nn$,
there exists a word $\infA$ over the alphabet $\{0,\ldots,k+1\}$ such that for each 
$i\le k$ and $v\in\Nn$ the letter $i+1$ appears in $\infA$ at least $v$ times if and only if $\forall x\!<\!v\, \exists y.\, \delta(i,x,y)$. 
In particular, $i+1$ appears in $\infA$ infinitely many times if and only if $\pi(i)$ holds.  
\end{lemma}

\begin{proof}
We reason in $\rca$. 
Given some $k \in \Nn$, we define a function
$C$ with domain $\{0,\ldots,k\} \times \Nn$ by letting $C(i,w) = \max\big\{v\leq w\mid \forall x\!<\!v\, \exists y\!<\! w.\, \delta(i,x,y)\big\}$
for $i \leq k$ and $w \in \Nn$. Clearly the function $C$ is computable and so exists by $\Delta^0_1$-comprehension.

Given some computable enumeration\footnote{$(n, k) \mapsto \frac{(n + k + 1)(n + k)}{2} + k$ is a simple enough example.} of pairs $\fun{\lbracket \cdot , \cdot \rbracket}{\Nn^2}{\Nn}$ that is monotone with respect to the coordinatewise order on $\Nn^2$, define the infinite word $\infA$ by:
\[\infA(j) =
\begin{cases}
i+1 & \text{if $j=\lbracket i, w \rbracket$, $i\leq k$,}\\
  & \text{ and $C(i,w)>\big|\{w'<w\mid\infA(\lbracket i, w'\rbracket)=i+1\}\big|$,}\\
0 & \text{otherwise}.
\end{cases}
\]
Again, $\infA(j)$ is computable so $\alpha$ can be obtained by $\Delta^0_1$-comprehension. Note that $\alpha\lbracket i', w \rbracket = i+1$ implies $i' = i$ for any $i,i'$. We now verify that $\infA$ satisfies the requirements
of the lemma.

First assume that $\forall x\!<\!v\, \exists y.\, \delta(i,x,y)$ holds for some $i\leq k$ and $v\in\Nn$. By $\Sigma^0_1$-collection, there exists some $w$ such that $\forall x\!<\!v \, \exists y\!<\!w.\, \delta(i, x, y)$. Let $\ell=\big|\{w'<w\mid\infA(\lbracket i, w'\rbracket)=i+1\}\big|$. If $\ell\geq v$ then we are done. Assume the contrary and notice that $C(i,w)\geq v$. This means that for $w'=w,w+1,\ldots,w+v-\ell-1$ we have $\infA(\lbracket i, w'\rbracket)=i+1$ (we use $\Sigma^0_1\ind$ to prove this). In total this gives us $v$ positions of $\infA$ that are labelled by $i+1$.

Now assume that there are at least $v$ positions of $\infA$ labelled by $i+1$. Let $w_0$ be the minimal position such that $\big|\{w'\leq w_0\mid\infA(\lbracket i, w'\rbracket)=i+1\}\big|= v$.
We know that $\infA(\lbracket i, w_0\rbracket)=i+1$ and that the set $\{w'<w_0\mid\infA(\lbracket i, w'\rbracket)=i+1\}$ has $v-1$ elements. This means that $C(i,w_0)\geq v$. By the definition of $C(i,w)$, it follows that $\forall x\!<\!v\, \exists y.\, \delta(i,x,y)$ holds.
\end{proof}

To complete the proof of Proposition~\ref{thm:dec-to-ind}, we will use Lemma~\ref{lem:pi-2-complete-word} to show that if the depth\=/$5$ fragment 
of $\mso(\Nn, {\le})$ is decidable, then Lemma~\ref{lem:induction-made-simple} 
can be applied to a~sequence of \mso sentences $(\psi_k)_{k \in \Nn}$ where $\psi_k$ basically says  ``$\Pi^0_2$ induction holds up to $k$''.

\begin{proofof}{Proposition~\ref{thm:dec-to-ind}}
For $k\in\Nn$, let $\psi_k$ be the $\mso$ sentence ``for every infinite word over the alphabet $\{0,\ldots,k\}$ there is a maximal letter $i\in\{0,\ldots,k\}$ occurring infinitely often''. More formally, $\psi_k$ is defined to be the depth\=/$5$ sentence
\begin{multline}
\forall X_0\, \forall X_1\, \ldots \forall X_k\, \left[\forall x\, \left( \bigvee_{i\le k} x\in X_i \land \bigwedge_{i<j\le k} \lnot\big(x\in X_i\land x\in X_j\big)\right)\right.\Longrightarrow\label{eq:formula}\\
\left.
\bigvee_{i \le k} \left((\forall x \, \exists y\! \ge\! x.\, y\in X_i) \wedge \bigwedge_{i<j\le k} (\exists x \, \forall y\!\ge\! x.\, y\notin X_j)\right)\right].\nonumber
\end{multline}
Clearly, $\rca$ proves that $\psi_0 \in \mso(\Nn,{\le})$ and for every $k\in\Nn$, if $\psi_k \in \mso(\Nn,{\le})$, then $\psi_{k+1} \in \mso(\Nn,{\le})$. So, by Lemma~\ref{lem:induction-made-simple} and the assumption on decidability of depth\=/$5$ $\mso(\Nn, {\le})$, each sentence $\psi_k$ is true in 
$(\Nn,{\le})$.

Now consider a $\Pi^0_2$ formula $\pi(i)$, possibly with parameters. Let $k \in \Nn$ and assume that $\pi(0)$ but $\neg \pi(k)$. Let $\alpha$ be
the word corresponding to $\pi$ and $k$ provided by Lemma~\ref{lem:pi-2-complete-word}. 
Since the MSO sentence $\psi_{k+1}$ is true in $(\Nn,{\le})$, 
there is a maximal letter $i$ appearing in $\alpha$ infinitely often. Clearly $0 < i < k+1$ and 
$\pi(i-1)$ but $\neg \pi(i)$.

Since $\pi(i)$ was an arbitrary $\Pi^0_2$ formula, we have proved $\Pi^0_2\ind$ and thus also $\Sigma^0_2\ind$.
\end{proofof}


\section{Making complementation ineffective}
\label{sec:two-compl}

The work of Sections~\ref{sec:ind-to-ramsey}--\ref{sec:dec-to-ind} proves the equivalence of items \ref{it:induction}, \ref{it:ramsey}, 
\ref{it:compl_buchi} and \ref{it:dec_mso} of Theorem \ref{thm:main}. However, \ref{it:compl_buchi}, concerning 
complementation of B\"uchi automata, contains an effectivity condition, namely that there exists an algorithm that produces 
an automaton complementing any given input automaton $\Aa$. 
It is natural to ask whether this effectivity condition 
can be dropped without compromising the logical strength of the statement. 

Below, we prove that the answer is positive, and therefore also item \refComplBuchiPrime~of Theorem \ref{thm:main}
is equivalent to the others. Our argument relies on the ideas of Section \ref{sec:dec-to-ind} and is similar in spirit to the one used in the proof of \cite[Theorem 3.1, (2) $\to$ (3)]{km:2016}, though somewhat simpler. Clearly \ref{it:compl_buchi} implies \refComplBuchiPrime. Hence, it is enough to show for instance that \refComplBuchiPrime~implies \ref{it:induction}:
\begin{proposition}
Provably in $\rca$, if for every nondeterministic B\"uchi automaton $\Aa$ there exists a B\"uchi automaton $\Bb$ such that for every infinite word $\infA$, 
$\Bb$ accepts $\infA$ exactly if $\Aa$ does not accept $\infA$, then $\Sigma^0_2\ind$ holds.  
\end{proposition}

\begin{proof}
Assume $\Sigma^0_2\ind$ fails and let $\pi(i)$ be a $\Pi^0_2$ formula such that $\pi(0)$ and $\pi(i) \to \pi(i+1)$ for each $i$, but $\neg \pi(k)$ for some $k$.
By Lemma \ref{lem:pi-2-complete-word} this means that there is a word $\alpha$ over the alphabet $\{0,\ldots,k+1\}$ such that there is no
maximal letter $i \le k+1$ appearing infinitely often in $\alpha$.

Consider the following B\"uchi automaton $\Aa$ working over $\{0,\ldots,k+1\}$: at some point, $\Aa$ nondeterministically chooses a letter $i$ and verifies that from that point onwards, $i$ appears infinitely many times but no $j > i$ appears at all. Apply complementation to obtain an automaton $\Bb$ which accepts exactly if $\Aa$ rejects.

Note that $\Aa$ rejects the word $\alpha$, because no matter when it makes its nondeterministic choice and what letter $i$ it chooses, 
either $i$ will appear only finitely many times or some $j>i$ will appear after the choice is made. Therefore, $\Bb$ has an accepting run on some word, namely on $\alpha$. By a standard application of the (finite) pigeonhole principle
$\ell + p$, it chooses the maximal letter occurring as one of $\beta(\ell), \ldots, \beta(\ell + p-1)$. This contradicts the assumption that $\Bb$ accepts exactly if $\Aa$ rejects.
\end{proof}

\section{Additive Ramsey and Ordered Ramsey imply \texorpdfstring{$\Sigma^0_2\ind$}{IND}}
\label{ap:ramsey-to-ind}

In this section, we give a direct proof showing that both Additive Ramsey's Theorem and Ordered Ramsey's Theorem imply $\Sigma^0_2\ind$. The implication from Additive Ramsey already follows from Theorem~\ref{thm:main}. However, the argument below is very simple 
and establishes a direct link between our Ramsey-theoretic statements and the induction scheme, without the detour through automata and MSO; thus, we feel it is worth including.

\begin{proposition}
\label{lem:ramsey-to-ind}
Over $\rca$, both Additive Ramsey's Theorem and Ordered Ramsey's Theorem imply $\Sigma^0_2\ind$.
\end{proposition}

\begin{proof}
By Lemma~\ref{lem:pi-2-complete-word}, to derive $\Sigma^0_2\ind$ it is enough to show that for every $k\in\Nn$ and every infinite word $\infA\in\{0,\ldots,k\}^\Nn$, there is a maximal letter $i$ appearing infinitely many times in $\infA$. Fix $k$ and $\infA$ and consider the colouring $\colA$ with values in $\{0,\ldots,k\}$ defined for $i<j$ as follows:
\[\colA(i,j)=\max\{\infA(\ell)\mid i\leq \ell <j\}.\]
The colouring $\colA$ can be viewed both as an additive colouring of $[\Nn]^2$ by elements of the semigroup $(\{0,\ldots,k\},\max)$, or as an ordered colouring w.r.t.\ the inverse of the usual order on $\{0,\ldots,k\}$. Thus, we can use either Additive Ramsey's Theorem or Ordered Ramsey's Theorem to obtain an infinite homoneous set $I$ for $\colA$. Let $i\in\{0,\ldots,k\}$ be the colour of $I$. By the definition of $\colA$, $i$ is the largest colour that appears infinitely many times in~$\infA$.
\end{proof}

\section{Additive Ramsey and Wilke algebras}
\label{sec:ind-to-wilke}

Wilke algebras provide an algebraic framework~\cite{wilke_algebraic,PinPerrin} for studying regular languages of infinite words. The crucial result concerning Wilke algebras is the extension principle (see Corollary~2.6 in~\cite{wilke_algebraic} and Theorem~5.1 in~\cite{PinPerrin}), saying that each semigroup morphism taking finite words over a~given alphabet into a finite Wilke algebra has a unique extension to all infinite words.
This extension property essentially expresses Additive Ramsey's Theorem, modulo some algebraic manipulation on Wilke algebras.
In the following section we formally state this fact in terms of reverse mathematics, by proving that over $\rca$ the Wilke extension principle is equivalent to Additive Ramsey's Theorem. Thus, the correct behaviour of yet another standard model of recognition of regular languages of infinite words turns out to be equivalent to $\Sigma^0_2\ind$.

\newcommand{\Sp}{S_{{+}}}
\newcommand{\Si}{S_{{\infpow}}}

We start by introducing the basic notions, for reference see~\cite{wilke_algebraic}. A Wilke algebra (originally called right binoid, see~\cite{wilke_algebraic}) is a pair $(\Sp,\Si)$ with operations: an associative operation ${\cdot_{{+}}}$ on $\Sp$; a left semigroup action ${\cdot_{{\infpow}}}$ of $\Sp$ on $\Si$; and a mapping $.^\infpow$ from $\Sp$ to $\Si$. Additionally, a Wilke algebra must satisfy the natural axioms denoted (PU), (RO) in~\cite{wilke_algebraic}.
The \emph{rotation law} (RO) states that, for every $x,y \in \Sp$, $(xy)^\infpow = x \cdot_\infpow (yx)^\infpow$. The \emph{pumping law} (PU) states that, for every $x \in \Sp$ and $n \ge 1$, $(x^n)^\infpow = x^\infpow$. Finite Wilke algebras (i.e., where both $\Sp$ and $\Si$ are finite) are the counterpart to finite\=/state B\"uchi automata in terms of recognisability of languages over infinite words.



The Wilke extension principle, as stated in~\cite[Corollary~2.6]{wilke_algebraic}, says that for every alphabet $\Sigma$ and finite Wilke algebra $(\Sp,\Si)$, each semigroup homomorphism $\fun{f}{\Sigma^{{+}}}{\Sp}$ has a unique extension to a homomorphism of Wilke algebras $(\Sigma^{{+}},\Sigma^\Nn)\to(\Sp,\Si)$ that is additionally \emph{Ramsey} (see Theorem~1.2.(C) in~\cite{wilke_algebraic}): $f(U_0)^\infpow=f(U_0U_1\cdots)$ for every sequence of finite words with $f(U_0)=f(U_1)=\ldots$ 

As $L_2$ is not expressive enough to manipulate such morphisms $f$ as first-class objects, we need to reformulate the extension principle in $L_2$. First, we say that a decomposition of an infinite word $\infA\in \Sigma^\Nn$ info a sequence of non-empty words $\infA=U_0 U_1\cdots$ is \emph{weakly Ramseyan w.r.t.\ $\fun{f}{\Sigma^{{+}}}{\Sp}$} if $f(U_1)=f(U_2)=\ldots$. If additionally $f(U_0)\cdot f(U_1) = f(U_0)$ and $f(U_1)\cdot f(U_2)=f(U_1)$ then we say that the decomposition is \emph{Ramseyan}. 
Essentially by the definition, Additive Ramsey's Theorem says that every infinite word admits a~particular kind of a weakly Ramseyan decomposition.

\begin{definition}
\label{def:wilke-ext}
The \emph{Wilke extension principle} says: for every alphabet $\Sigma$ and finite Wilke algebra $(\Sp,\Si)$ with a semigroup homomorphism $\fun{f}{\Sigma^{{+}}}{\Sp}$, for every $\infA\in \Sigma^\Nn$, there exists a~unique value $x\in \Si$ such that for every weakly Ramseyan decomposition $\infA=U_0 U_1\cdots$ of $\infA$ w.r.t.\ $f$, we have
\begin{equation}
\label{eq:ramseyan}
x = f(U_0)\cdot \big(f(U_1)\big)^\infpow.
\end{equation}
The above defined value is called \emph{the value} of the decomposition $U_0,U_1,\ldots$
\end{definition}

Now, let us state the main theorem of this section.

\begin{theorem}
\label{thm:wilke-and-ramsey}
Over $\rca$, the Wilke extension principle is equivalent to Additive Ramsey's Theorem, and hence also to \ref{it:induction}--\ref{it:dec_mso} of Theorem \ref{thm:main}.
\end{theorem}

Before we move on, first notice the following two simple lemmas.

\begin{lemma}
\label{lem:exists-ramsey}
Over $\rca$, the Wilke extension principle implies that if $f$ is a semigroup homomorphism $\fun{f}{\Sigma^{{+}}}{\Sp}$ then for every $\infA$ there exists at least one weakly Ramseyan factorisation of $\infA$ w.r.t.\ $f$.
\end{lemma}

\begin{proof}
Given a semigroup $\Sp$ and a homomorphism $f$, let $(\Sp,\Si)$ be the Wilke algebra with $\Si=\{0,1\}$, $r^\infpow=0$ for every $r\in\Sp$, and $r\cdot_{\infpow} s= s$ for every $r\in\Sp$ and $s\in\Si$. Now assume that for some $\infA\in \Sigma^\Nn$ there is no weakly Ramseyan decomposition w.r.t.\ $f$. In that case both values $x=0$ and $x=1$ satisfy~\eqref{eq:ramseyan}, which violates the condition of uniqueness in the Wilke extension principle.
\end{proof}

\begin{lemma}
\label{lem:weakly-to-ramseyan}
Over $\rca$, if $\infA=U_0U_1\cdots$ is a weakly Ramseyan decomposition of $\infA$ w.r.t. $f$ then there exists a Ramseyan decomposition of $\infA$ w.r.t. $f$ of the same value.
\end{lemma}

\begin{proof}
Consider a weakly Ramseyan decomposition $\infA=U_0U_1\cdots$ Let $N\eqdef|\Sp|!$ be the factorial of the size of the semigroup. It is known that for every $r\in\Sp$ we have $r^N=r^N\cdot r^N$. Now it is enough to group the decomposition into blocks of size $N$, i.e.: $V_0=U_0U_1\cdots U_N$ and $V_K=U_{K\cdot N + 1}U_{K\cdot N+2}\cdots U_{K\cdot N+N}$ for $K=1,2,\ldots$. The new decomposition is also weakly Ramseyan with $f(V_1)=f(V_2)=\ldots=f(U_1)^N$. Additionally, we know that $f(V_0)\cdot f(V_1)=f(V_0)$ and $f(V_1)\cdot f(V_2)=f(V_1)$. The value of the new decomposition is the same as the value of the original one by the axioms (PU) and (RO) of Wilke algebras.
\end{proof}

\begin{corollary}
Over $\rca$, the Wilke extension principle implies Additive Ramsey's Theorem.
\end{corollary}

\begin{proof}
Consider an additive colouring $\fun{C}{[\Nn]^2}{S}$, with values in a finite semigroup $S$. Let $\infA\in S^\Nn$ be the infinite word defined as $\infA(k)=C(k,k{+}1)$ and let $\fun{f}{S^{{+}}}{S}$ be the product homomorphism: $f(s_1s_2\cdots s_k)=s_1s_2\cdots s_k$. It is easy to see that for $i<j\in\Nn$ we have $C(i,j)=\infA(i)\infA(i+1)\cdots\infA(j-1)=f\big(\infA(i)\infA(i+1)\cdots\infA(j-1)\big)$.

Apply Lemma~\ref{lem:exists-ramsey} to obtain a~weakly Ramseyan decomposition $\infA=U_0U_1\cdots$ Lemma~\ref{lem:weakly-to-ramseyan} shows that there must also exist a~Ramseyan decomposition of $\infA=U_0'U_1'\cdots$. Such a decomposition induces a homogeneous set for $C$, given as $I=\{|U_0'|, |U_0'U_1'|,\ldots\}$.
\end{proof}

What remains is to prove the opposite implication of Theorem~\ref{thm:wilke-and-ramsey}, as expressed by the following lemma.

\begin{lemma}
Over $\rca$, Additive Ramsey's Theorem implies the Wilke extension principle.
\end{lemma}

\begin{proof}
The proof of this lemma follows the standard strategy of proving the Wilke extension principle. However, we need to pay attention to make sure that the argument can be done in $\rca$.

Consider a~Wilke algebra $(\Sp,\Si)$, a~homomorphism $\fun{f}{\Sigma^{{+}}}{\Sp}$, and an infinite word $\infA$. Let $C$ be the~colouring defined by \[C(i,j)\eqdef f\big(\infA(i)\infA(i+1)\cdots\infA(j-1)\big)\] for $i<j\in\Nn$. By Additive Ramsey's Theorem, there exists an infinite set $I=\{i_0<i_1<\ldots\}$ homogeneous for $C$. This set provides a~Ramseyan decomposition $\infA=U_0'U_1'\cdots$, with $U'_{k+1}=\infA(i_k)\cdots\infA(i_{k+1}{-1})$. Let $x_0$ be the value of this decomposition. We claim that $x_0$ satisfies~\eqref{eq:ramseyan} for every weakly Ramseyan decomposition of $\infA$.

Consider any weakly Ramseyan decomposition $\infA=U_0U_1\cdots$ By Lemma~\ref{lem:weakly-to-ramseyan} we can assume that the decomposition is in fact Ramseyan. Let $s=f(U_0)$, $s'=f(U_0')$, $e=f(U_1)$, $e'=f(U_1')$. We need to prove that $se^\infpow=s'(e')^\infpow$. But, by the decompositions $(U_n)_{n\in\Nn}$ and $(U_n')_{n\in\Nn}$, we know that $s'=sx$ for some $x\in\Sp$ 
(notice that $s = s'y$ implies that there is $x$ such that $s'= s'e'=s'yx=sx$), 
$xe'$ is $\Rr$-equivalent to $e$ and $e$ is $\Jj$-equivalent to $e'$ ($\Rr$ and $\Jj$ are Green relations, see Definition~\ref{def:green}). Now, Proposition~2.7 from~\cite{PinPerrin}, which is proved using finite combinatorics, shows that the pairs $(s,e)$ and $(s',e')$ are \emph{conjugate}, meaning that there exist $x,y\in \Sp$ such that $e=xy$, $e'=yx$ and $s'=sx$. In that case, using the axioms of Wilke algebras, we get that:
\[
  s e^\infpow = s (xy)^\infpow= s x (yx) ^ \infpow = s' (e')^\infpow.
  \tag*{\qedhere}
\]
\end{proof}

\subsection*{Relation to Ramsey homomorphisms} As discussed above, the original statement of the Wilke extension principle from~\cite{wilke_algebraic} cannot be expressed in $L_2$. For the sake of this section, we proposed a simplified version of this principle (Definition~\ref{def:wilke-ext}). The rest of this section is devoted to informally arguing that the Wilke extension principle in fact easily implies the original principle. Let us assume the Wilke extension principle. Take a homomorphism $\fun{f}{\Sigma^{{+}}}{\Sp}$. Consider $\fun{f}{\Sigma^\Nn}{\Si}$ defined as $f(\infA)=x$ from the definition of the Wilke extension principle. Clearly $f$ is Ramseyan in the sense of~\cite{wilke_algebraic}. We need to argue that $f$ is a homomorphism and that it is unique. First notice that for every $w\in \Sigma^{{+}}$ we have $f(www\cdots) = f(w)^\infpow$ by~\eqref{eq:ramseyan}. Now consider $\beta=w\cdot \infA$ for a finite word $w\in \Sigma^{{+}}$. By Lemma~\ref{lem:exists-ramsey} and Lemma~\ref{lem:weakly-to-ramseyan} (both provable in $\rca$), the Wilke extension principle implies that $\infA= U_0U_1\cdots$ for a Ramseyan decomposition. In that case $\beta = w U_0U_1\cdots$ and $f(\beta)=f(w)f(U_0)f(U_1)^\infpow=f(w)\cdot f(\infA)$. Thus, $f$ is a homomorphism.

Now assume that $\fun{f'}{\Sigma^\Nn}{\Si}$ is another Ramseyan homomorphism. We need to prove that $f'=f$. Consider an infinite word $\infA$. Notice that since $f'$ is Ramseyan, the value $x=f'(\infA)$ satisfies~\eqref{eq:ramseyan}. Thus, by the uniqueness of $x$ we know that $f'(\infA)=f(\infA)$.
%

\begin{remark}
Murakami et al.~\cite{murakami-et-al:ramseyan-factor} consider a statement they call the \emph{weak Ramseyan factorisation theorem}, which says essentially that every infinite word has a weak Ramseyan decomposition w.r.t.~any finite colouring, i.e.~without the additivity requirement. They show that already the weak factorisation theorem for two colours implies a set existence principle known as $\mathsf{ADS}$, which is unprovable in $\rca + \bigcup_{n\in \omega} \Sigma^0_n\ind$. This is another confirmation of the crucial role of the additivity condition in the Ramsey\=/like statements used in automata theory.  
\end{remark}


\section{\texorpdfstring{$\Sigma^0_2\ind$}{IND} implies Bounded-width K\"onig}
\label{sec:ind-to-konig}

%
%
\newcommand{\thmIndToKonig}{Over $\rca$, $\Sigma^0_2\ind$ implies Bounded-width K\"onig's Lemma (see Definition~\ref{def:bounded-konig}).}

\begin{theorem}
\label{thm:ind-to-konig}
\thmIndToKonig
\end{theorem}

Simpson and Yokoyama \cite{sy:very-weak-wkl} have independently studied various weak forms of Weak K\"onig's Lemma, including principles they call
$\mathsf{WKL}(\mathrm{w}$-$\mathrm{bd})$ and $\mathsf{WKL}(\mathrm{ext}$-$\mathrm{bd})$ that can be seen to be equivalent to Bounded-width K\"onig's Lemma over $\rca$. They also prove that $\Sigma^0_2\ind$ implies these principles, and have some results on circumstances under which the implication reverses (it cannot reverse in general due to the incomparability of $\Sigma^0_2\ind$ and $\wkl$). 

\begin{proofof}{Theorem \ref{thm:ind-to-konig}}
Let us fix a graph $G$ with vertices contained in $Q\times \Nn$ for some finite set $Q$. The usual way of proving K\"onig's Lemma would start by defining the subset $G'$ of those vertices $v$ of $G$ for which the subgraph under $v$ is infinite. Having defined $G'$, we could inductively pick any infinite path in $G'$ and---assuming $G$ does in fact contain arbitrarily long finite paths starting in $Q \times \{0\}$---we are guaranteed not to get stuck. The issue is whether we can obtain $G'$ by $\Delta^0_1$-comprehension.

A $\Pi^0_1$-definition of $G'$ is provided by a standard trick used in the context of $\wkl$. Notice that for every fixed $k$ there can be at most $|Q|$ vertices of $G$ of the form $(q,k)$. Thus a vertex $(q,k)$ is in $G'$ if and only if it has the $\Pi^0_1$ property that for every $\ell\geq k$ there exists a vertex $(q',\ell)$ reachable from $(q,k)$ by a path in $G$; here the existential quantifier over $(q',\ell)$ is bounded in terms of $\ell$ and $|Q|$. 

What remains is to give a $\Sigma^0_1$-definition of $G'$.

Consider two numbers $k<\ell$ and a vertex $v=(q,k)$ of $G$. We will say that $v$ \emph{dies before $\ell$} if there is no path in $G$ from $v$ that reaches a vertex of the form $(q',\ell)$. For $i=0,1,\ldots,|Q|$ we will say that \emph{$i$ vertices die infinitely many times} if
\begin{align*}
\forall j\, \exists k\!>\!j\, \exists \ell\!>\!k.\ \text{there are at least $i$ vertices of the form $(q,k)$}\\
\text{that die before $\ell$.}
\end{align*}

Notice that the property of $i$ that \emph{$i$ vertices die infinitely many times} is $\Pi^0_2$. Clearly if $i \le i'$ and \emph{$i'$ vertices die infinitely many times} then \emph{$i$ vertices die infinitely many times}. By $\Sigma^0_2\ind$ we can fix $i_0$ as the maximal $i$ such that \emph{$i$ vertices die infinitely many times}. By the definition, if $i > i_0$ then there exists $j(i)$ such that for every $\ell>k>j(i)$ there are fewer than $i$ vertices of the form $(q,k)$ that die before $\ell$. Notice that we can assume $j_0 := j(i_0 +1)$ to be an upper bound for all $j(i)$ where $i>i_0$. This means that for $\ell>k>j_0$ we have at most $i_0$ vertices of the form $(q,k)$ that die before $\ell$. Additionally, for infinitely many $k$ there is $\ell>k$ such that exactly $i_0$ vertices of the form $(q,k)$ die before $\ell$. The following claim shows how one can find a witness that the subgraph under a vertex $v$ is infinite.

\begin{claim}
\label{clm:if-omits-then-infinite}
Assume that we are given $\ell>k>j_0$ and a vertex $v=(q,k)$ such that exactly $i_0$ vertices of the form $(q',k)$ with $q'\neq q$ die before $\ell$. Then the subgraph under $v$ is infinite.
\end{claim}

\begin{proof}
Assume to the contrary that for some $\ell'>\ell$ there is no vertex of the form $(q',\ell')$ that can be reached from $(q,k)$ by a path in $G$. This means that $(q,k)$ dies before $\ell'$. Therefore, there are at least $i_0+1$ vertices of the form $(q',k)$ that die before $\ell'$. This contradicts the way $j_0$ was chosen.
\end{proof}

Clearly, if for some $\ell>k$ and a vertex $v=(q,k)$ we know that $v$ dies before $\ell$ then the subgraph of $G$ under $v$ is finite.

We shall now use Claim~\ref{clm:if-omits-then-infinite} to give a $\Sigma^0_1$-definition of $G'$. We will say that $v=(q,k)$ belongs to $G'$ if there exist $\ell>k'>\max(k,j_0)$ and $i_0$ vertices of the form $(q',k')$ such that all of them die before $\ell$ and some other vertex of the form $(q'',k')$ is reachable in $G$ by a path from $v$. Clearly this is a $\Sigma^0_1$-definition. It remains to prove that it defines $G'$. First assume that $v$ satisfies the above property and fix $\ell$, $k'$, and $(q'',k')$ as in the definition. By Claim~\ref{clm:if-omits-then-infinite} we know that the subgraph under $(q'',k')$ is infinite. Since $(q'',k')$ is reachable from $v$ in $G$, this implies that also the subgraph under $v$ is infinite and thus $v\in G'$. Now assume that $v=(q,k)\in G'$. By the choice of $i_0$ we know that there exist $\ell>k'>\max(k,j_0)$ and exactly $i_0$ vertices of the form $(q',k')$ that die before $\ell$. Since the subgraph under $v$ is infinite, we know that some vertex of the form $(p,\ell)$ is reachable from $v$ in $G$. Notice that any path connecting $v$ and $(p,\ell)$ needs to contain a vertex of the form $(q'',k')$. Clearly $(q'',k')$ cannot be among the $i_0$ vertices that die before $\ell$. Thus $v$ satisfies the above condition.

We have thus shown that the graph $G'$ is indeed $\Delta^0_1$-definable, so we can use it to complete the proof. 
Let the vertex $(q,0)$ of $G$ satisfy the hypothesis of Bounded-width K\"onig's Lemma. 
Clearly, $(q,0) \in G'$. Just as clearly, each $v=(q,k)\in G'$ is connected by an edge to some $(q',k+1)\in G'$.
This lets us define an infinite path in $G'$ by primitive recursion. 
Let $\patA(0)$ be $(q,0)$. 
If $\patA(k)$ is defined let $\patA(k+1)=(q',k+1)$ for the  minimal $q' \! \in \! Q$ such that $(q',k+1)\in G'$ and there is an edge in $G$ between $\patA(k)$ and $(q',k+1)$. 
By the construction $\patA$ is an infinite path in $G'$, and hence in $G$, starting in $(q,0)$.
\end{proofof}

\section{\texorpdfstring{$\Sigma^0_2\ind$}{IND} implies determinisation}
\label{sec:ind-to-det}

The entirety of this section is devoted to a proof of the following theorem, which coincides with implication $\ref{it:induction} \to \ref{it:det_buchi}$ of Theorem \ref{thm:main}.

\begin{theorem}
\label{thm:ind-to-det}
Over $\rca$, $\Sigma^0_2\ind$ implies the existence of an algorithm which, given a nondeterministic B\"uchi automaton $\Bb$ over an alphabet $\Sigma$, outputs an equivalent deterministic Rabin automaton $\Aa$
over the same alphabet 
such that for every $\infA \in \Sigma^\Nn$ 
we have
\[\text{$\Aa$ accepts $\infA$}\quad\Longleftrightarrow\quad\text{$\Bb$ accepts $\infA$.}\]
\end{theorem}

The proof scheme presented here is based on a determinisation procedure proposed in \cite{muller_schupp} (see~\cite{Thomas,Wilke} for similar arguments and a comparison of this determinisation method to the method of Safra). Our exposition follows lecture notes of Boja{\'n}czyk~\cite{Boj}. 
Although the general structure of the argument is standard, we need to take additional care to ensure that the reasoning can be conducted in $\rca$ using only $\Sigma^0_2\ind$.

\subsection{Transducers}
\label{ssec:transducers}

The proof of Theorem~\ref{thm:ind-to-det} will be split into separate steps that will allow us to successively simplify the objects under consideration. 
The steps typically take the form of lemmas stating the existence of automata with certain properties. All the lemmas in the remainder of the section are asserted to hold provably in $\rca + \Sigma^0_2\ind$. Moreover, all automata whose existence is claimed can be obtained effectively given a nondeterministic  B\"uchi automaton $\Bb$ over the alphabet $\Sigma$ and possibly other automata mentioned in the hypothesis of each particular lemma.

To merge the steps we will use the notion of a deterministic transducer that transforms one infinite word into another. 

\begin{definition}
A transducer is a deterministic finite automaton, without accepting states, where each transition is additionally labelled by a letter from some \emph{output alphabet}. More formally, a transducer with an input alphabet $\Sigma$ and an output alphabet $\Gamma$ is a tuple $\Tt=\langle Q, q_\init, \delta\rangle$ where $q_\init\in Q$ is an initial state and $\fun{\delta}{Q\times \Sigma}{\Gamma\times Q}$. 
\end{definition}
A transducer naturally defines a function $\fun{\Tt}{\Sigma^\Nn}{\Gamma^\Nn}$. Formally, such a function is a third-order object and thus not available in second-order arithmetic. However, given a word $\alpha$, we can use $\Delta^0_1$-comprehension to obtain the unique infinite word produced by $\Tt$ on the input $\infA$. Whenever we write $\Tt(\infA)$, we have this word in mind.

It is easy to see that a transducer can be used to reduce the question of acceptance from one deterministic automaton to another, as stated by the following lemma.
\begin{lemma}
\label{lem:transducers}
For every deterministic Rabin automaton $\Aa$ with input alphabet $\Gamma$ and every transducer $\fun{\Tt}{\Sigma^\Nn}{\Gamma^\Nn}$, there exists a deterministic Rabin automaton $\Aa\circ\Tt$ which accepts an infinite word $\infA \in \Sigma^\Nn$ if and only if $\Aa$ accepts $\Tt(\infA)$.
\end{lemma}

\begin{proof}
Let the set of states of $\Aa\circ\Tt$ be the product of the states of $\Aa$ and the states of $\Tt$. The transition function of $\Aa\circ\Tt$ follows both the transitions of $\Tt$ and the transitions of $\Aa$ over letters output by $\Tt$:
\[\delta^{\Aa\circ\Tt}\big((q^\Aa,q^\Tt),a\big)=(\delta^\Aa(q^\Aa,b), q')\quad\text{where $\delta^\Tt(q^\Tt,a)=(b,q')$}.\]
The Rabin acceptance condition of $\Aa\circ\Tt$ is taken to be the acceptance of $\Aa$, skipping the second coordinate of the states. Clearly 
the first coordinate of the run of $\Aa\circ\Tt$ over an infinite word $\infA$ equals the run of $\Aa$ over $\Tt(\infA)$, so $\Aa\circ\Tt$ accepts  $\infA$ if and only if $\Aa$ accepts $\Tt(\infA)$.
\end{proof}

\subsection{\texorpdfstring{$Q$}{Q}-dags}
\label{ssec:q-dags}

In the exposition below we will work with infinite words representing the set of all possible runs of a nondeterministic automaton over a fixed infinite word. Let us define a $Q$-dag to be a directed acyclic graph where the set of nodes is $Q \times \Nn$ and every edge is of the form  
\[((q,k), (p,k+1)) \qquad \text{for some $p,q \in Q$ and $k \in \Nn$.}\]
Furthermore, every edge is coloured by one of the two colours: ``accepting'' or ``non-accepting''. We assume that there are no parallel edges. A path in a $Q$-dag is a finite or infinite sequence of nodes connected by edges (either accepting or non-accepting). As we will see, we can assume that every $Q$-dag is rooted---there is a distinguished element $q_\init\in Q$ such that all the edges of the $Q$-dag lie on a path that starts in the vertex $(q_\init,0)$. We call a vertex $(q,k)$ \emph{reachable} if there is a path from $(q_\init,0)$ to $(q,k)$ in $\infA$. We say that an infinite path in a $Q$-dag is \emph{accepting} if it starts in $(q_\init,0)$ and contains infinitely many accepting edges. 

\begin{figure}
\centering
\begin{tikzpicture}[scale=0.55]
\qDagFigure
\end{tikzpicture}
\caption{A $Q$-dag and a single letter from the alphabet $[Q]$. The accepting edges are represented by solid lines, and non-accepting edges are dashed lines.}
\label{fig:Q-dag-ap}
\end{figure}
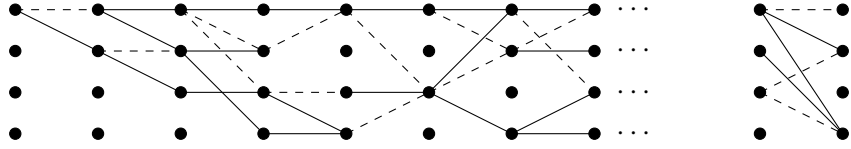

Every $Q$-dag can be naturally represented as an infinite word, where the $k$-th letter encodes the set of edges of the form $((q,k),(q',k+1))$. The alphabet used for this purpose will be the set of transition matrices $[Q]$ defined in Section~\ref{sec:intro_automata}. An example of a $Q$-dag and a letter in $[Q]$ are depicted on Figure~\ref{fig:Q-dag-ap}.

We will be particularly interested in $Q$-dags that are \emph{tree\=/shaped}. A $Q$-dag is \emph{tree-shaped} if every node $(q,k)$ has at most one incoming edge (i.e.~an edge from a node of the form $(p,k-1)$. Notice that it makes sense to say that a letter $M\in[Q]$ is tree-shaped and a $Q$-dag is tree-shaped if and only if all of its letters are tree-shaped. Figure~\ref{fig:tree-shaped-Q-dag} depicts a tree-shaped $Q$-dag.

A $Q$-dag is \emph{infinite} if for every $k$ there exists a path connecting the root $(q_\init,0)$ with a vertex of the form $(q',k)$. Similarly, a $Q$-dag is \emph{infinite under $(q,k)$} if for every $k'\geq k$ there exists a path connecting the vertex $(q,k)$ with a vertex of the form $(q',k')$.

\begin{figure}
\centering
\begin{tikzpicture}[scale=0.55]
\treeShapedQDag
\end{tikzpicture}
\caption{A tree-shaped $Q$-dag.}
\label{fig:tree-shaped-Q-dag}
\end{figure}
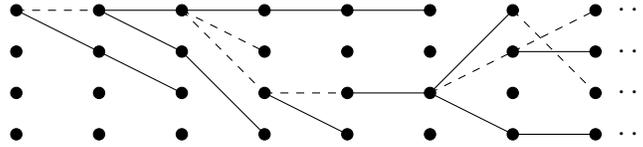

\begin{lemma}
\label{lem:to-dags}
Given a nondeterministic B\"uchi automaton $\Bb$ over an~alphabet~$\Sigma$, there exists a transducer $\Tt_1$ that takes as input an infinite word $\infA\in\Sigma^\Nn$ and outputs a $Q$-dag $\Tt_1(\infA)$ such that $\Bb$ accepts $\infA$ if and only if $\Tt_1(\infA)$ contains an accepting path.
\end{lemma}

\begin{proof}
The transducer $\Tt_1$, after reading a finite word $w\in\Sigma^\ast$, stores in its state the set of states of $\Bb$ reachable from $q_\init^\Bb$ over $w$. The initial state of $\Tt_1$ is $\{q_\init\}$. Given a state $R\subseteq Q$ of $\Tt_1$ and a letter $a$, the transducer moves to the state
\[R'=\{q'\mid (q,a,q')\in\delta^\Bb, q\in R\}\]
and outputs a letter $M\in [Q]$ such that $M(q,q')=M_a(q,q')$ if $q\in R$ and $M(q,q')=\trNon$ if $q\notin R$ (see Section~\ref{sec:intro_automata} for the definition of $M_a$ and $[Q]$). Clearly there is a computable bijection between the accepting runs of $\Bb$ over $\infA$ and accepting paths in the $Q$-dag $\Tt_1(\infA)$.
\end{proof}

\subsection{Reduction to tree-shaped \texorpdfstring{$Q$}{Q}-dags}
\label{ssec:reduction-to-trees}

The next lemma shows that one can use a transducer to reduce general $Q$-dags to tree-shaped $Q$-dags
.

\begin{lemma}
\label{lem:henryk}
There exists a transducer $\Tt_2$ that takes as input a $Q$-dag $\infA'$ and outputs a tree-shaped $Q$-dag $\Tt_2(\infA')$ such that $\infA'$ contains an accepting path if and only if $\Tt_2(\infA')$ contains an accepting path.
\end{lemma}

To prove this lemma we will use a lexicographic order on paths in a given $Q$-dag. A crucial ingredient here is Bounded-width K\"onig's Lemma from Section~\ref{sec:ind-to-konig}. Additionally, we need to make sure that the graph to which Bounded-width K\"onig's Lemma is applied can be obtained using $\Delta^0_1$-comprehension. For this purpose we use $\Sigma^0_2\ind$ once again.

In the proof we will use the following definition.

\begin{definition}[Profiles]
For a finite path $w$ in a $Q$-dag, define its profile to be the word over the alphabet $\{\trYes,\trAcc\}\times Q^2$ which is obtained by replacing each edge $((q,k),(q',k+1))$ in $w$ by $(x,q,q')$ where $x\in\{\trYes,\trAcc\}$ is the type of the edge ($\trAcc$ for accepting and $\trYes$ for non-accepting). Let us fix any linear order ${\preceq}$ on $\{\trYes,\trAcc\}\times Q^2$ such that $(\trAcc,q,q')\prec(\trYes,p,p')$. Let ${\preceq}$ be the lexicographic order on paths induced by the order ${\preceq}$ on their profiles. We call a path $w$ \emph{optimal} if it is lexicographically minimal among all paths with the same source and target.
\end{definition}

Lemma~\ref{lem:henryk} follows from Claims~\ref{claim_one} and~\ref{claim_two}.

\begin{claim}
\label{claim_one}
There is a transducer $\fun{\Tt}{[Q]^\Nn}{[Q]^\Nn}$ such that if the input is $\infA$ then $\Tt(\infA)$ is tree-shaped with the same reachable vertices as in $\infA$, and such that every finite path from the root in $\Tt(\infA)$ is an optimal path in $\infA$.
\end{claim}

\begin{proof}
We start with the following observation about the order ${\preceq}$. Let $w,w',u,u'$ be paths in a $Q$-dag $\infA$ such that the target of $w$ (resp.~$u$) is the source of $w'$ (resp.~$u'$); and $w$, $u$ are of equal length. Then $ww'\preceq uu'$ if and only if $w\prec u$ or $w=u$ and $w'\preceq u'$.

Now let us define $\Tt(\infA)$ by choosing, for every vertex reachable in $\infA$, an ingoing edge that belongs to some optimal path. Putting all of these edges together will yield a tree-shaped $Q$-dag as in the statement of the claim. To produce such edges, after reading the first $k$ letters, the automaton keeps in its state a linear order on $Q$ that corresponds to the lexicographic ordering on the optimal paths leading from the root to the nodes at depth $k$. 
Updating the order on $Q$ upon reading a~new letter from $[Q]$ is possible thanks to the observation above---thus, only finitely many states that keep the current order on $Q$ are enough.
\end{proof}

Notice that the above proof is purely constructive and the statement of Claim~\ref{claim_one} involves only finite combinatorics, therefore it can be performed in $\rca$.

\begin{claim}
\label{claim_two}
Let $\Tt$ be the transducer from Claim~\ref{claim_one}. 
If the input $\infA$ to $\Tt$ contains an accepting path then so does the output $\Tt(\infA)$.
\end{claim}

The rest of this subsection is devoted to a proof of Claim~\ref{claim_two}. Let $\infA$ be an input to $\Tt$. Assume that $\patA\in (Q\times \Nn)^\Nn$ is a path that contains infinitely many accepting edges in $\infA$. A node $v$ in the $Q$-dag $\infA$ is said to be \emph{$\patA$\=/merging} if there exists a finite path in $\Tt(\infA)$ that leads from $v$ to a vertex on $\patA$. Our aim is to define the following set of vertices in $Q\times\Nn$:
\[t=\{v\in Q\times\Nn\mid \text{$v$ is $\patA$-merging}\}.\]
The above definition is clearly a $\Sigma^0_1$-definition of $t$.

\begin{subclaim}
There exists a $\Pi^0_1$ predicate over vertices $v$ equivalent to ``$v$ is $\patA$-merging''. As a consequence, $t$ is definable by $\Delta^0_1$-comprehension.
\end{subclaim}

The proof of this subclaim makes essential use of $\Sigma^0_2\ind$ and is similar to the proof of Theorem~\ref{thm:ind-to-konig}.

\begin{proof}
For $i=0,1,\ldots,|Q|$ we will say that $i$ is \emph{$\patA$\=/merging infinitely often} if 
\[\forall j\, \exists k\!>\!j.\ \text{there are at least $i$ $\patA$-merging vertices of the form $(q,k)$ in $\Tt(\infA)$}.\]
The above property of $i$ is clearly a $\Pi^0_2$ property. Let $i_0$ be the maximal $i\leq|Q|$ that is $\patA$\=/merging infinitely often. Such $i_0$ exists by $\Sigma^0_2\ind$. Clearly if $i\leq i'$ and $i'$ is $\patA$\=/merging infinitely often then $i$ is also $\patA$\=/merging infinitely often. By the definition, if $i > i_0$ then there exists $j(i)$ such that for all $k>j(i)$ there are fewer than $i$ $\patA$-merging vertices of the form $(q,k)$ in $\Tt(\infA)$. Notice that we can assume $j_0 := j(i_0 +1)$ to be an upper bound for all $j(i)$ where $i_0<i\leq |Q|$. This means that if $k>j_0$ then there are at most $i_0$ $\patA$-merging vertices of the form $(q,k)$ in $\Tt(\infA)$.

We can now provide a $\Pi^0_1$-definition of $t$ (actually a $\Sigma^0_1$-definition of the vertices outside $t$). A vertex $v=(q,k)$ does not belong to $t$ if $(\star)$: there exists $k'>\max(k,j_0)$ and $i_0$ vertices of the form $v_0=(q_0,k')$, $v_1=(q_1,k')$, \ldots, $v_{i_0}=(q_{i_0},k')$ such that:
\begin{itemize}
\item all the vertices $v_0,\ldots,v_{i_0}$ are $\patA$\=/merging in $\Tt(\infA)$,
\item no path from $v$ to any of $v_i$ for $i=0,1,\ldots,i_0$ exists,
\item there is no path in $\Tt(\infA)$ from $v$ to a vertex of the form $(q',\ell)$ that lies on $\patA$ with $\ell \leq k'$.
\end{itemize}
The latter two conditions are decidable, while the first one is $\Sigma^0_1$. In total, the condition $(\star)$ is $\Sigma^0_1$.

We will now prove that the negation of $(\star)$ in fact defines $t$. First assume that $v=(q,k)\notin t$. Recall that there are infinitely many $k'$ such that there are exactly $i_0$ $\patA$-merging vertices of the form $(q',k')$ in $\Tt(\infA)$. In particular, there exists $k'>\max(k,j_0)$ and $i_0$ vertices of the form $v_0=(q_0,k')$, $v_1=(q_1,k')$, \ldots, $v_{i_0}=(q_{i_0},k')$ such that all of them are $\patA$\=/merging. Since $v$ is not $\patA$-merging, there cannot be a path from $v$ to any of the vertices $v_i$ for $i=1,2,\ldots,i_0$. Similarly, there cannot be a path from $v$ to $\patA$. Therefore, $v$ satisfies $(\star)$.

On the other hand, assume that $v$ has the property $(\star)$ as witnessed by some $k'$ and vertices $v_0,\ldots,v_{i_0}$. Assume to the contrary that $v$ is $\patA$-merging. Let this be witnessed by a path $w$ from $v$ to a vertex $v''=(q'',k'')$ on $\patA$. By the last item of $(\star)$, we must have $k''>k'$. Let $p\in Q$ be the state such that $(p,k')$ lies on the path $w$. Clearly $(p,k')$ is $\patA$-merging so it needs to be one of the vertices $v_1,\ldots,v_{i_0}$. But in that case this vertex can be reached from $v$ by a path in $\Tt(\infA)$, a contradiction.
\end{proof}

We can now apply Bounded-width K\"onig's Lemma (see Definition~\ref{def:bounded-konig}) to the graph with set of vertices $t$ and with edges inherited from $\Tt(\infA)$. This graph has arbitrarily long finite paths starting in $(q_\init,0)$, because each vertex on $\pi$ belongs to $t$ and is reachable
from $(q_\init,0)$ by a path in $\Tt(\infA)$ contained within $t$. We obtain an infinite path $\patA'$ in $\Tt(\infA)$ contained within $t$. Our aim is to prove that $\patA'$ contains infinitely many accepting edges. Assume to the contrary that for some $k\in\Nn$ there is no accepting edge of the form $((p,\ell),(p',\ell+1))$ for $\ell>k$ on $\patA'$. Let $(p,k)$ be a vertex that belongs to $\patA'\cap Q\times\{k\}$. Since $\patA'$ is a path in $t$, we know that $(p,k)$ is $\patA$\=/merging. Let $w$ be a path witnessing this fact and let $(p',k')$ be its final vertex, which lies on $\patA$. Since $\patA$ is accepting, we know that it contains an accepting edge of the form $((r,\ell),(r',\ell+1))$ with $k<\ell$. Let $(q,\ell+1)$ be a vertex that belongs to $\patA'\cap Q\times\{\ell+1\}$. As in the case of $(p,k)$, we have a path $w'$ witnessing that $(q,\ell+1)$ is $\patA$\=/merging, which reaches $\patA$ in a vertex $(q',\ell')$.

This means that in 
$\infA$
there are two paths between $(p,k)$ and $(q',\ell')$ (see Figure~\ref{fig:il-claim-two}): the first one follows $w$ and $\patA$, the second one follows $\patA'$ and $w'$. Notice that the latter path is contained in $t$. This means that the profile of the path through $\patA'$ and $w'$ is smaller than the profile of the path through $w$ and $\patA$. By the definition of the order on profiles, since there is an accepting edge on the respective fragment of $\patA$, the corresponding fragment of the path $\patA'$ needs to contain an accepting edge. This contradicts the assumption that there is no accepting edge of the form $((p,k''),(p',k''+1))$ for $k''>k$ on $\patA'$.

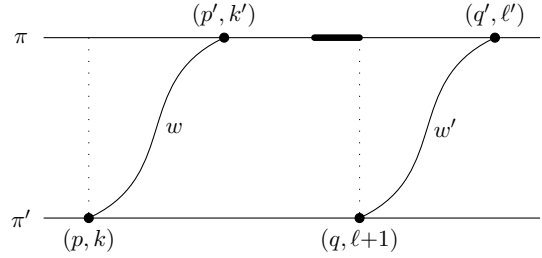
\begin{figure}
\centering
\begin{tikzpicture}[scale=0.6,every node/.style={scale=0.8}]
\draw (0,0) -- (11,0);
\draw (0,-4) -- (11,-4);

\draw[loosely dotted] (1,0) -- (1, -4);

\draw[loosely dotted] (7,0) -- (7, -4);

\node[dot] (lb) at (1,-4) {};
\node at (1,-4.5) {$(p,k)$};

\node[dot] (lu) at (4,-0) {};
\node at (4,+0.5) {$(p',k')$};

\draw (lb) .. controls ++(2,1) and ++(-2,-1) .. (lu);
\node at (2.9,-2) {$w$};

\draw[line width=2.5, cap=round] (6,0) -- (7,0);

\node[dot] (rb) at (7,-4) {};
\node at (7,-4.5) {$(q,\ell{+}1)$};

\node[dot] (ru) at (10,-0) {};
\node at (10,+0.5) {$(q',\ell')$};

\draw (rb) .. controls ++(2,1) and ++(-2,-1) .. (ru);
\node at (8.9,-2) {$w'$};

\node at (-0.5,0)  {$\patA$};
\node at (-0.5,-4) {$\patA'$};
\end{tikzpicture}
\caption{An illustration to the proof of Claim~\ref{claim_two}. The upper horizontal line is the path $\patA$ in $\infA$ that may not be a path in $\Tt(\infA)$. The paths $w$ and $w'$ witness that $(p,k)$ and $(q,\ell{+}1)$ are both $\patA$\=/merging. The boldfaced part of $\patA$ is the chosen accepting edge that appears on $\patA$. Among the two paths from $(p,k)$ to $(q',\ell')$: one through $w$ and the other through $w'$; the latter belongs to $\Tt(\infA)$. Therefore, it has to have smaller profile than the former, in particular it has to contain an accepting edge in between the vertices $(p,k)$ and $(q,\ell+1)$.}
\label{fig:il-claim-two}
\end{figure}

This concludes the proof of Claim~\ref{claim_two} and thus of Lemma~\ref{lem:henryk}.

\subsection{Recognising accepting tree-shaped \texorpdfstring{$Q$}{Q}-dags}
\label{ssec:recognising-tree-shapes}

The proof of Theorem~\ref{thm:ind-to-det} is concluded by the following lemma and an application of Lemma~\ref{lem:transducers}.

\begin{lemma}
\label{lem:michal}
There exists a deterministic Rabin automaton $\Aa$ over the alphabet $[Q]$ that for every tree-shaped $Q$-dag $\infA''\in [Q]^\Nn$ accepts it if and only if $\infA''$ contains an accepting path.
\end{lemma}

We will start by defining the states and transitions of the constructed Rabin automaton. Then we will prove that it in fact verifies if a given infinite word that is a tree-shaped $Q$-dag contains an accepting path.


In general, the size of the constructed Rabin automaton is one of the crucial parameters of the construction, as it influences the running time of the algorithms for verification and synthesis of reactive systems. However, in this work we are mainly focused on the fact that an equivalent deterministic automaton exists. Therefore, the relatively simple 
construction presented here will be far from optimal. For a discussion on optimality of the constructions involved, see~\cite{colcombet_safra}. We conjecture that soundness of more optimal determinisation procedures, such as \emph{Safra's construction}~\cite{safra_determinisation}, may be proven in $\Sigma^0_2\ind$. 

\begin{definition}
Fix a finite nonempty set $Q$. We will say that $\tau$ is a \emph{$Q$-scheme} if $\tau$ is a finite tree with:
\begin{itemize}
\item internal nodes labelled by $Q$,
\item leaves uniquely labelled by $Q$,
\item edges uniquely labelled by $\{0,1,\ldots,2\cdot|Q|\}$, these labels are called \emph{identifiers},
\item each edge additionally marked as either ``accepting'' or ``non-accepting''.
\end{itemize}
Additionally, the root cannot be a~leaf and every node of $\tau$ that is neither the root nor a~leaf has to have at least two children.
\end{definition}

Notice that we are not requiring a~$Q$-scheme to be balanced as a~tree. It is easy to see that since the leaves of $\tau$ are uniquely labelled by $Q$, $\tau$ has at most $2\cdot|Q|$ nodes. Therefore, the requirement that the edge labels from $\{0,\ldots,2\cdot|Q|\}$ need to be pairwise distinct is not restrictive. Clearly the number of $Q$-schemes is finite (in fact exponential in $|Q|$). Let the set of states of $\Aa$ be the set of all $Q$-schemes. Let the initial state of $\Aa$ be the $Q$-scheme consisting of two nodes: the root and its only child, both labelled by $q_\init$. Let the edge between the root and the unique leaf be labelled by the identifier $0$ and be ``non-accepting''.

We will now proceed to the definition of the transitions of $\Aa$. Assume that the automaton is in a state $\tau$ and reads a tree-shaped letter $M\in[Q]$, see Figure~\ref{fig:scheme-transition}.

\begin{figure}
\centering
\begin{tikzpicture}[scale=0.55]    
  \path[acc] (0*\dxs, -0) edge node[above, idf] {$2$} ++(\dxs, +1.5);
  \path[tra] (0*\dxs, -0) edge node[below, idf] {$3$} ++(\dxs, -1.0);
  
  \path[tra] (1*\dxs, +1.5) edge node[above, idf] {$9$} ++(\dxs, +0.5);
  \path[acc] (1*\dxs, +1.5) edge node[below, idf] {$6$} ++(\dxs, -0.5);
  \path[tra] (1*\dxs, -1.0) edge node[above, idf] {$5$} ++(\dxs, +1.0);
  \path[tra] (1*\dxs, -1.0) edge node[below, idf] {$7$} ++(\dxs, +0.0);
  \path[acc] (1*\dxs, -1.0) edge node[below, idf] {$1$} ++(\dxs, -1.0);
    
  \node[dot] at (0*\dxs, 0) {};
  \node[dot] at (1*\dxs, +1.5) {};
  \node[dot] at (1*\dxs, -1.0) {};
  
  \foreach \y in {-2,...,2} {
    \node[dot] at (2*\dxs, \y) {};
  }
    
  \foreach \x in {0,...,1} {
    \foreach \y in {-3,...,3} {
      \node[dot] at (3*\dxs+\x*\dxs, \y) {};
    }
  }

  \path[tra] (3*\dxs, +3.0) -- ++(\dxs, -0.0);
  \path[acc] (3*\dxs, +3.0) -- ++(\dxs, -1.0);
    
  \path[tra] (3*\dxs, +1.0) -- ++(\dxs, -0.0);
  \path[acc] (3*\dxs, +1.0) -- ++(\dxs, -1.0);
  
  \path[acc] (3*\dxs, -1.0) -- ++(\dxs, -0.0);
  \path[acc] (3*\dxs, -2.0) -- ++(\dxs, -0.0);
  \path[tra] (3*\dxs, -2.0) -- ++(\dxs, -1.0);
  
\end{tikzpicture}
\caption{A $Q$-scheme $\tau$ (a state of $\Aa$) and a tree-shaped letter $M\in[Q]$ encountered by $\Aa$. The ``non-accepting'' edges in $\tau$ are dashed. The leaves of $\tau$ are arranged according to some fixed order on $Q$ in such a way as to match the layout of $M\in[Q]$. To simplify the picture we do not include the states in $Q$ labeling the nodes of $\tau$, using dots instead.}
\label{fig:scheme-transition}
\end{figure}
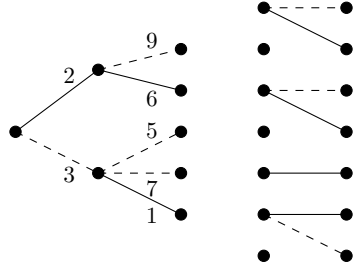

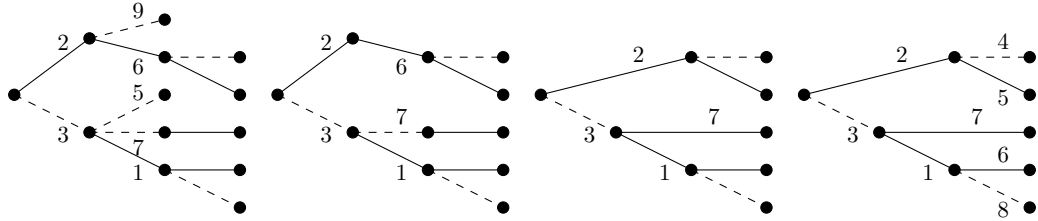
\begin{figure}
\centering
\begin{tikzpicture}[scale=0.5]    

  \coordinate (sero) at (0,0);
  
  \path[acc] ($(sero)+(0*\dxs, -0)$) edge node[above, idf] {$2$} ++(\dxs, +1.5);
  \path[tra] ($(sero)+(0*\dxs, -0)$) edge node[below, idf] {$3$} ++(\dxs, -1.0);
  
  \path[tra] ($(sero)+(1*\dxs, +1.5)$) edge node[above, idf] {$9$} ++(\dxs, +0.5);
  \path[acc] ($(sero)+(1*\dxs, +1.5)$) edge node[below, idf] {$6$} ++(\dxs, -0.5);
  \path[tra] ($(sero)+(1*\dxs, -1.0)$) edge node[above, idf] {$5$} ++(\dxs, +1.0);
  \path[tra] ($(sero)+(1*\dxs, -1.0)$) edge node[below, idf] {$7$} ++(\dxs, +0.0);
  \path[acc] ($(sero)+(1*\dxs, -1.0)$) edge node[below, idf] {$1$} ++(\dxs, -1.0);
       
  \node[dot] at ($(sero)+(0*\dxs, 0)$) {};
  \node[dot] at ($(sero)+(1*\dxs, +1.5)$) {};
  \node[dot] at ($(sero)+(1*\dxs, -1.0)$) {};
  
  \node[dot] at ($(sero)+(2*\dxs, 2)$) {};
  \node[dot] at ($(sero)+(2*\dxs, 1)$) {};
  \node[dot] at ($(sero)+(2*\dxs, 0)$) {};
  \node[dot] at ($(sero)+(2*\dxs, -1)$) {};
  \node[dot] at ($(sero)+(2*\dxs, -2)$) {};
  
  \foreach \y in {-3,...,1} {
    \node[dot] at ($(sero)+(3*\dxs, \y)$) {};
  }  
    
  \path[tra] ($(sero)+(2*\dxs, +1.0)$) -- ++(\dxs, -0.0);
  \path[acc] ($(sero)+(2*\dxs, +1.0)$) -- ++(\dxs, -1.0);
  
  \path[acc] ($(sero)+(2*\dxs, -1.0)$) -- ++(\dxs, -0.0);
  \path[acc] ($(sero)+(2*\dxs, -2.0)$) -- ++(\dxs, -0.0);
  \path[tra] ($(sero)+(2*\dxs, -2.0)$) -- ++(\dxs, -1.0);   

  \coordinate (sero) at (7,0);
  
  \path[acc] ($(sero)+(0*\dxs, -0)$) edge node[above, idf] {$2$} ++(\dxs, +1.5);
  \path[tra] ($(sero)+(0*\dxs, -0)$) edge node[below, idf] {$3$} ++(\dxs, -1.0);
  
  \path[acc] ($(sero)+(1*\dxs, +1.5)$) edge node[below, idf] {$6$} ++(\dxs, -0.5);
  \path[tra] ($(sero)+(1*\dxs, -1.0)$) edge node[above, idf] {$7$} ++(\dxs, +0.0);
  \path[acc] ($(sero)+(1*\dxs, -1.0)$) edge node[below, idf] {$1$} ++(\dxs, -1.0);
              
  \node[dot] at ($(sero)+(0*\dxs, 0)$) {};
  \node[dot] at ($(sero)+(1*\dxs, +1.5)$) {};
  \node[dot] at ($(sero)+(1*\dxs, -1.0)$) {};
  
  \node[dot] at ($(sero)+(2*\dxs, 1)$) {};
  \node[dot] at ($(sero)+(2*\dxs, -1)$) {};
  \node[dot] at ($(sero)+(2*\dxs, -2)$) {};
  
  \foreach \y in {-3,...,1} {
    \node[dot] at ($(sero)+(3*\dxs, \y)$) {};
  }  
    
  \path[tra] ($(sero)+(2*\dxs, +1.0)$) -- ++(\dxs, -0.0);
  \path[acc] ($(sero)+(2*\dxs, +1.0)$) -- ++(\dxs, -1.0);
  
  \path[acc] ($(sero)+(2*\dxs, -1.0)$) -- ++(\dxs, -0.0);
  \path[acc] ($(sero)+(2*\dxs, -2.0)$) -- ++(\dxs, -0.0);
  \path[tra] ($(sero)+(2*\dxs, -2.0)$) -- ++(\dxs, -1.0);  

  \coordinate (sero) at (14,0);
  
  \path[acc] ($(sero)+(0*\dxs, -0)$) edge node[above, idf] {$2$} ++(2*\dxs, +1.0);
  \path[tra] ($(sero)+(0*\dxs, -0)$) edge node[below, idf] {$3$} ++(\dxs, -1.0);
  
  \path[acc] ($(sero)+(1*\dxs, -1.0)$) edge node[above, idf] {$7$} ++(2*\dxs, +0.0);
  \path[acc] ($(sero)+(1*\dxs, -1.0)$) edge node[below, idf] {$1$} ++(\dxs, -1.0);
              
  \node[dot] at ($(sero)+(0*\dxs, 0)$) {};
  \node[dot] at ($(sero)+(1*\dxs, -1.0)$) {};
  
  \node[dot] at ($(sero)+(2*\dxs, 1)$) {};
  \node[dot] at ($(sero)+(2*\dxs, -2)$) {};
  
  \foreach \y in {-3,...,1} {
    \node[dot] at ($(sero)+(3*\dxs, \y)$) {};
  }  
    
  \path[tra] ($(sero)+(2*\dxs, +1.0)$) -- ++(\dxs, -0.0);
  \path[acc] ($(sero)+(2*\dxs, +1.0)$) -- ++(\dxs, -1.0);
  
  \path[acc] ($(sero)+(2*\dxs, -2.0)$) -- ++(\dxs, -0.0);
  \path[tra] ($(sero)+(2*\dxs, -2.0)$) -- ++(\dxs, -1.0);  

  \coordinate (sero) at (21,0);
  
  \path[acc] ($(sero)+(0*\dxs, -0)$) edge node[above, idf] {$2$} ++(2*\dxs, +1.0);
  \path[tra] ($(sero)+(0*\dxs, -0)$) edge node[below, idf] {$3$} ++(\dxs, -1.0);
  
  \path[acc] ($(sero)+(1*\dxs, -1.0)$) edge node[above, idf] {$7$} ++(2*\dxs, +0.0);
  \path[acc] ($(sero)+(1*\dxs, -1.0)$) edge node[below, idf] {$1$} ++(\dxs, -1.0);
       
  \node[dot] at ($(sero)+(0*\dxs, 0)$) {};
  \node[dot] at ($(sero)+(1*\dxs, -1.0)$) {};
  
  \node[dot] at ($(sero)+(2*\dxs, 1)$) {};
  \node[dot] at ($(sero)+(2*\dxs, -2)$) {};
  
  \foreach \y in {-3,...,1} {
    \node[dot] at ($(sero)+(3*\dxs, \y)$) {};
  }  
    
  \path[tra] ($(sero)+(2*\dxs, +1.0)$) edge node[above, idf] {$4$} ++(\dxs, -0.0);
  \path[acc] ($(sero)+(2*\dxs, +1.0)$) edge node[below, idf] {$5$} ++(\dxs, -1.0);
  
  \path[acc] ($(sero)+(2*\dxs, -2.0)$) edge node[above, idf] {$6$} ++(\dxs, -0.0);
  \path[tra] ($(sero)+(2*\dxs, -2.0)$) edge node[below, idf] {$8$} ++(\dxs, -1.0);  
\end{tikzpicture}
\caption{The successive transformations of the scheme $\tau$ when performing steps $1$ to $4$ of a transition of $\Aa$.}
\label{fig:scheme-transformation}
\end{figure}

The resulting state $\tau'$ is constructed by performing the following four steps depicted on Figure~\ref{fig:scheme-transformation}.

\begin{description}
\item[Step 1.] We append the new letter $M$ to the $Q$-scheme $\tau$ obtaining a new tree. The identifiers on the newly created edges are undefined and some nodes may have exactly one child. However, all the nodes are labelled by states in $Q$, either coming from $\tau$ or from $M$.
\item[Step 2.] We eliminate paths that die out before reaching the target states of $M$. In the running example, this means eliminating edges with identifiers $9$ and $5$.
\item[Step 3.] We eliminate unary nodes, thus joining several edges into a single edge. This means that a path which only passes through nodes of degree one gets collapsed into a single edge, the identifier for such an edge is inherited from the first (i.e.~leftmost) edge on the path. The newly created edge is ``accepting'' if and only if any of the collapsed edges were ``accepting''. In the running example, this means eliminating the unary nodes that are the targets of edges with identifiers $2$ and $7$.
\item[Step 4.] Finally, if there are edges that do not have identifiers, these edges get assigned arbitrary identifiers that are not currently used. In the running example we add identifiers $4$, $5$, $6$, and $8$.
\end{description}

This completes the definition of the state update function. We now define the acceptance condition. 

\subsection*{The acceptance condition.} When executing a transition, the automaton described above goes from one $Q$-scheme to another $Q$-scheme. For each identifier, a transition can have three possible effects, described below:

\begin{description}
\item[Delete] An edge can be deleted in Step~2 (it dies out) or in Step~3 (it is merged with a path to the left). The identifier of such an edge is said to be \emph{deleted} in the transition. The deleted identifiers in the running example are $9$, $5$, and $6$. Since we reuse identifiers, an identifier can still be present after a transition that deletes it, because it has been added again in Step~4. This happens to identifiers $5$ and $6$ in the running example.
\item[Refresh] In Step~3, an entire path with edges identified by $e_1,e_2, \cdots, e_k$ is folded into its first edge identified by $e_1$. If any of the edges identified by $e_2, \cdots, e_n$ was ``accepting'' then we say that the identifier $e_1$ is \emph{refreshed}. In the running example the refreshed identifiers are $2$ and $7$ (the edge identified by $2$ was already ``accepting'' while the edge identified by $7$ become ``accepting'' because of the merging).
\item[Nothing] An identifier might be neither \emph{deleted} nor \emph{refreshed}. In the running example, this is the case for identifiers $1$ and $3$.
\end{description}

The following lemma describes the key property of the above data structure.

\begin{lemma}
\label{lem:acceptance}
For every tree-shaped $Q$-dag $\infA\in[Q]^\Nn$, the following are equivalent:
\begin{enumerate}
\item $\infA$ contains an accepting path,
\item some identifier is deleted only finitely often but refreshed infinitely often.
\end{enumerate}
\end{lemma}

Before proving the above lemma, we show how it completes the proof of Lemma~\ref{lem:michal}. Clearly, the second condition above can be expressed as a Rabin condition on transitions of $\Aa$---the Rabin pairs $(E_i,F_i)$ range over the set of identifiers $i=1,\ldots,2\cdot|Q|$, a transition is in $E_i$ if an edge with the identifier $i$ is deleted and is in $F_i$ if the edge is refreshed.

\begin{proofof}{Lemma~\ref{lem:acceptance}}
First assume that $\infA$ contains an accepting path $\patA$. Let $\rho$ be the sequence of states of $\Aa$ when reading $\infA$. Notice that for every $k$, the path $\patA$ induces a path in the $Q$-scheme $\rho(k)$ that connects the root with a leaf labelled by a state $q(k)$ such that $\patA(k)=(q(k),k)$. Let $e^{(k)}_0,\ldots,e^{(k)}_{j(k)}$ be the identifiers of the edges on this path. Notice that $j(k)\leq |Q|$ because each internal node of a $Q$-scheme has at least two children and leaves of $Q$-schemes are uniquely labelled by the states in $Q$. We will say that a position $j=0,1,\ldots,|Q|$ is \emph{unstable} if for infinitely many $k$ either $j(k)<j$ or some identifier $e^{(k)}_{j'}$ for $j'\leq j$ is \emph{deleted} in the $k$-th transition in $\rho$. Notice that $0$ is \emph{stable} because we never delete the first edge of a $Q$-scheme. Let $j_0$ be the greatest \emph{stable} number; such a number exists by $\Sigma^0_2\ind$.

By $\Sigma^0_2$-collection we can find a number $k_0$ such that for $k\geq k_0$ we have $j(k)\geq j_0$ and no identifier $e^{(k)}_{j'}$ with $j'\leq j_0$ is \emph{deleted} in the $k$-th transition in $\rho$. Therefore, for every $j'\leq j_0$ and $k\geq k_0$ we have
\[e^{(k)}_{j'}=e^{(k_0)}_{j'}.\]

Let $i=e^{(k)}_{j_0}$. Clearly by the definition of $j_0$ we know that the identifier $i$ is not deleted for $k\geq k_0$. It remains to prove that $i$ is refreshed infinitely many times. Assume to the contrary that for some $k_1\geq k_0$ and every $k\geq k_1$ the identifier $i$ is never refreshed in the $k$-th transition in $\rho$. First notice that $\patA$ contains an accepting edge of the form $((q,k_2-1),(q',k_2))$ for some $k_2 \ge k_1$. The edge identified by $e^{(k_2)}_{j(k_2)}$ is accepting in $\rho(k_2)$---this is the last edge on the path corresponding to $\patA$ in the $Q$-scheme obtained after reading the $k_2$-th letter of $\infA$. There are two cases. If $j(k_2) = j_0$, then $i$ is refreshed in the $k_2$-th transition, contradicting our assumption that $i$ is not refreshed beyond $k_1$. Otherwise, $j(k_2) \ge j_0 + 1$ and, by the definition of $j_0$ we know that for some $k_3\geq k_2$ the identifier $e^{(k_3)}_{j_0+1}$ is deleted in the $k_3$-th transition in $\rho$. Notice that since $\patA$ is an infinite path, this identifier cannot be deleted in Step~2 as it never dies out.
Therefore, it must be the case that $e^{(k_3)}_{j_0+1}$ is deleted in Step~3 and that $j(k_3+1)=j_0$. Let us prove by $\Sigma^0_1\ind$ on $k=k_2,k_2+1,\ldots,k_3$ that either:
\begin{itemize}
\item the identifier $i$ is refreshed in the $k'$-th transition of $\rho$ for some $k'$ such that $k_2\leq k'\leq k$, or
\item there exists an accepting edge in the $Q$-scheme $\rho(k)$ that is identified by $e^{(k)}_{j'}$ for some $j'$ such that $j_0<j'\leq j(k)$.
\end{itemize}

For $k=k_2$ the second possibility holds. The inductive step follows directly from the definition of the transitions of $\Aa$---an accepting edge propagates to the left, firing successive refreshes for the merged identifiers. For $k=k_3$ we know that there is no $j'$ such that $j_0<j'\leq j(n)$ thus the first possibility needs to hold. This contradicts our assumption that there was no refresh on $i$ after the $k_1$\=/th letter of $\infA$ was read. This concludes the proof of the first implication in Lemma~\ref{lem:acceptance}.

Now assume that $\infA$ is a tree-shaped $Q$-dag accepted by the automaton $\Aa$. Let us fix the run $\rho$ of $\Aa$ over $\infA$ and assume that $i_0$ is an identifier that is deleted only finitely many times but refreshed infinitely many times. Let $k_0$ be such that the identifier $i_0$ is never deleted after the $k_0$-th transition of $\Aa$. Our aim is to prove that the $Q$-dag $\infA$ contains an accepting path. 

We start by noticing that for every $k\geq 0$ and an edge identified by $e$ in the $Q$-scheme $\rho(k)$, this edge corresponds to a finite path $w_{k,e}$ in the $Q$-dag $\infA$. For the newly created edges that are assigned new identifiers in Step~4, the corresponding path is an edge $(q,k),(q',k')$ from the letter $M$. For edges that were assigned an identifier earlier, the path is defined inductively, by merging the paths whenever we merge edges in Step~3. Using $\Sigma^0_1\ind$ we easily prove that a corresponding edge is marked ``accepting'' if and only if the path contains an accepting edge in $\infA$. If an identifier $i$ is refreshed then the path gets longer and contains at least one new accepting transition.

In this way, we can track the path corresponding to the edges identified by $i_0$ for $k\geq k_0$. Since the identifier $i_0$ is refreshed infinitely many times, the path corresponding to it is prolonged infinitely many times. Notice that the source of the paths corresponding to $i_0$ is fixed and of the form $(q(k_0),k_0)$---the identifier $i_0$ is never merged to the left. Clearly, to every $k\geq k_0$ we can effectively assign a state $q(k)$ such that for some $k'>k_0$ the path $w_{k',i_0}$ passes through $(q(k),k)$---such $k'$ exists because $i$ is refreshed infinitely many times. This gives us a $\Delta^0_1$-definition of an infinite path $\patA$ that starts in $(q(k_0),k_0)$. We can append it to a path from $(q_\init, 0)$ to $(q(k_0),k_0)$ 
and obtain a path $\patA'$ starting in $(q_\init, 0)$. Notice that each refresh of $i_0$ corresponds to a new accepting edge on $\patA$, 
which means that $\patA'$ is accepting.
\end{proofof}


\section{Conclusions and further work}
\label{sec:conclusion}

In this work we have characterised the logical strength of B\"uchi's decidability theorem and related results over the theory $\rca$. We proved over $\rca$ that complementation for B\"uchi automata is equivalent to $\Sigma^0_2\ind$, as is the decidability of $\mso(\Nn,{\le})$ (to the extent that this can be expressed). 

Many concepts related to automata on infinite words are $\Sigma^0_2$ or $\Pi^0_2$ and thus potentially problematic without $\Sigma^0_2\ind$---which is needed, for instance, to have access to the set of states occurring infinitely often in a run or to make sense of some more sophisticated acceptance conditions alternative to B\"uchi's. The picture suggested by our work is that, on the one hand, $\Sigma^0_2\ind$ is indeed necessary for the theory of infinite word automata to behave reasonably, but on the other hand, this minimal reasonability condition already suffices to prove all the basic results. 
This situation is completely different for automata on infinite trees, where the concepts also make sense already in $\rca + \Sigma^0_2\ind$, but proving the complementation theorem or decidability of \mso requires much more~\cite{km:2016}.

We are thus led to the general question whether the entire theory of automata on infinite words requires exactly $\rca + \Sigma^0_2\ind$. This includes in particular the following issues:
\begin{itemize}
\item Does McNaughton's determinisation theorem imply $\Sigma^0_2\ind$ over $\rca$?
\item What about developing the  Wagner hierarchy (see~\cite[Chapter~V.6]{PinPerrin})?
\item Does $\rca + \Sigma^0_2\ind$ prove the uniformisation theorem for automata, in the form: for a given automaton $\Aa$ over the alphabet
$\{0,1\}^2$ such that $\forall X\,\exists Y\, (\text{$\Aa$ accepts $X\otimes Y$})$, there 
exists an automaton $\Bb$ such that $\forall X\,\exists! Y\, (\text{both $\Aa$ and $\Bb$ accept $X\otimes Y$})$ (see~\cite[Theorem~27]{rabinovich_decidable})?
\end{itemize}


\bibliographystyle{alpha}
\bibliography{bibliography}

\end{document}